\documentclass[11pt]{article}
\usepackage{amsmath,amsthm}
\usepackage{amssymb,latexsym}
\usepackage{float}
\usepackage{mathpazo}

\newtheorem{theorem}{Theorem}[section]

\newtheorem{algorithm}[theorem]{Algorithm}
\newtheorem{claim}[theorem]{Claim}
\newtheorem{lemma}[theorem]{Lemma}
\newtheorem{conjecture}[theorem]{Conjecture}

\newtheorem{corollary}[theorem]{Corollary}

\theoremstyle{definition}
\newtheorem{definition}[theorem]{Definition}
\newtheorem{remark}[theorem]{Remark}

\floatstyle{ruled}
\newfloat{program}{h}{loa}
\floatname{program}{SDP}
\newfloat{problem}{h}{loa}
\floatname{problem}{Problem}
\newfloat{algorithm}{h}{loa} 
\floatname{algorithm}{Algorithm}
\newfloat{measurement}{h}{loa}
\floatname{measurement}{Measurement}

\newcommand{\qedsymb}{\hfill{\rule{2mm}{2mm}}}
\renewenvironment{proof}[1][]{\begin{trivlist} 
\item[\hspace{\labelsep}{\bf\noindent Proof#1:\/}] }{\qedsymb\end{trivlist}}

\def\calH{{\cal H}}

\def\R{\mathbb{R}}

\def\sdp{{\rm{sdp}}}

\def\id{I}

\def\CC{\mathbb{C}}
\def\RR{\mathbb{R}}

\def\bprod{\otimes_b}

\newcommand\ip[1]{{\left\langle {#1} \right\rangle}}
\newcommand\ipb[1]{{\big\langle {#1} \big\rangle}}
\newcommand\ipB[1]{{\Big\langle {#1} \Big\rangle}}

\newcommand\ket[1]{{ |{#1} \rangle }}
\newcommand\bra[1]{{ \langle {#1} | }}

\newcommand{\kbbig}[1]{\Bigl| #1  \Bigr\rangle\Bigl\langle #1\Bigr|}

\newcommand{\NP}{\textsc{NP}}

\newcommand{\MAXCUT}{\textsc{MaxCut}}
\newcommand{\MAXKCUT}{\textsc{Max-k-Cut}}
\newcommand{\ELIN}{\textsc{E3LIN2}}

\newcommand{\SETCOVER}{\textsc{SetCover}}
\newcommand{\MAXCLIQUE}{\textsc{MaxClique}}

\newcommand{\eps}{\varepsilon}
\renewcommand{\epsilon}{\varepsilon}

\begin{document}

\title{\bf Unique Games with Entangled Provers are Easy}

\author{Julia Kempe\footnote{School of Computer Science, Tel-Aviv University, Tel-Aviv 69978, Israel. Supported by
the European Commission under the Integrated Project Qubit Applications (QAP) funded by the IST
directorate as Contract Number 015848, by an Alon Fellowship of the Israeli Higher Council of
Academic Research, by an Individual Research Grant of the Israeli Science Foundation, by a European
Research Council (ERC) Starting Grant and by a Raymond and Beverly Sackler Career Development
Chair.}
 \and
 Oded Regev\footnote{School of Computer Science, Tel-Aviv University, Tel-Aviv 69978, Israel. Supported
   by the Binational Science Foundation, by the Israel Science Foundation,
   by the European Commission under the Integrated Project QAP funded by the IST directorate as Contract Number 015848
   and by a European Research Council (ERC) Starting Grant.}
\and
Ben Toner\thanks{School of Physics, The University of Melbourne,
   Victoria 3010, Australia. Part of this work was completed at CWI
   (Amsterdam). Supported by the Dutch BSIK/BRICKS project, by the
  European Commission under the Integrated Project QAP funded by the
  IST directorate as Contract Number 015848, and by NWO VICI project 639-023-302.}
}


\maketitle

\begin{abstract}
We consider one-round games between a classical verifier and two provers who share entanglement. We show that
when the constraints enforced by the verifier are `unique' constraints (i.e., permutations), the value of the
game can be well approximated by a semidefinite program. Essentially the only algorithm known previously was
for the special case of binary answers, as follows from the work of Tsirelson in 1980. Among other things,
our result implies that the variant of the unique games conjecture where we allow the provers to share
entanglement is false. Our proof is based on a novel `quantum rounding technique', showing how to take a
solution to an SDP and transform it to a strategy for entangled provers. Using our approximation by a
semidefinite program we also show a parallel repetition theorem for unique entangled games.
\end{abstract}

\section{Introduction}

\paragraph{Games:}
For nearly two decades, two-prover one-round games have played a major role in
many of the most important developments in theoretical computer science.
Such games consist of a verifier and two provers who are unable to communicate with each other.
The game starts when the verifier sends two questions, one to each prover,
chosen according to some joint distribution. Each prover then replies
with an answer chosen from the alphabet $\{1,\ldots,k\}$ for some $k \ge 1$.
Finally, the verifier decides whether to accept or reject, based on the
answers he received. The {\em value} of such a game is defined as the
maximum success probability that the provers can achieve. For example,
let us consider the following very simple game known as the CHSH game~\cite{Clauser:69a}:
the verifier sends a random bit to each of the provers, who then reply with
one bit each (so $k=2$). The verifier accepts if and only if the XOR of the answers
is equal to the AND of his questions. A moment's reflection shows
that the value of this game is $\frac{3}{4}$, and is obtained, say,
when the provers always return $0$.

One of the most important breakthroughs in theoretical computer science was
the discovery of the PCP theorem in the early 90s~\cite{AS98,ALMSS98}.
Combined with Raz's parallel repetition theorem~\cite{Raz98},
it implies the following.

\begin{theorem}[\cite{AS98,ALMSS98,Raz98}]
\label{thm:pcp}
For any $\delta > 0$ there exists a $k = k(\delta)$ such that it is \NP-hard to
determine whether, given a (two-prover one-round) game with answers from a domain of size $k$,
its value is $1$ or at most $\delta$.
\end{theorem}

This result has led to many important advances in the field, including in particular
many tight NP-hardness results. For instance, H{\aa}stad~\cite{Hastad01}
showed that it is NP-hard to tell whether a given 3SAT formula is satisfiable,
or not more than a $\frac{7}{8}+\eps$ fraction of its constraints can be satisfied.
This shows that the algorithm that simply assigns random values
to the variables is essentially optimal. Other tight NP-hardness results
that follow from the PCP theorem include a hardness factor of $\frac{1}{2}+\eps$
for $\ELIN$~\cite{Hastad01}, a hardness factor of $n^{1-\eps}$ for $\MAXCLIQUE$~\cite{Hastad-clique},
and a hardness factor of $\ln n$ for $\SETCOVER$~\cite{Feige98}.

One important special case of games is that of {\em unique games}. Here,
the verifier's decision is restricted to be of a very specific form.
Namely, for any questions $s,t$ sent to the provers, the verifier
accepts answers $a,b$ if and only if $b = \sigma_{st}(a)$ where
$\sigma_{st}$ is some permutation on $\{1,\ldots,k\}$.
In 2002, Khot~\cite{Khot02} presented a conjecture known as the
unique games conjecture (UGC) that essentially says that
it is hard to approximate the value of a unique game, even
if we are only interested in distinguishing the almost satisfiable
case from the almost completely unsatisfiable case.
\begin{conjecture}[Unique games conjecture~\cite{Khot02}]
\label{conjecture:ugc}
For any $\eps,\delta > 0$ there exists a $k = k(\eps,\delta)$ such that it is \NP-hard to
determine whether, given a unique game with answers from a domain of size $k$,
its value is at least $1-\eps$ or at most $\delta$.
\end{conjecture}
It is not hard to see that determining whether the value of a unique game is $1$
(i.e., the game is perfectly satisfiable) can be done efficiently using a simple
algorithm, and therefore it is crucial that we insist here on $\eps>0$
(cf.\ Theorem~\ref{thm:pcp}). Let us also mention that there {\em exist}
$\eps,\delta>0$ for which the problem in the conjecture {\em is} known to be
NP-hard (even with $k=2$). This follows from H{\aa}stad's hardness result for $\MAXCUT$~\cite{Hastad01}.
Despite a considerable amount of work in the last few years,
the plausibility of the conjecture is still uncertain, and this issue
is currently one of the central topics in theoretical computer science.

The tremendous importance of the unique games conjecture stems from the fact
that for many fundamental problems, it implies strong, and often tight,
inapproximability results that are not known to hold under more conventional assumptions.
As an example, let us consider the $\MAXCUT$ problem. The best known algorithm for this problem was given in 1994 by Goemans and Williamson,
and achieves an approximation factor of $\approx 0.878$~\cite{GoemansW95}.
It consists of two main steps: first, one writes a semidefinite programming (SDP)
relaxation of the given $\MAXCUT$ instance, where by `relaxation' we mean that by construction,
the value of the SDP is guaranteed to be not smaller than the size of the maximum cut.
This SDP can then be solved efficiently using known techniques for convex optimization,
such as the ellipsoid algorithm (see, e.g., \cite{BoydV04}). The second part of their algorithm is a `rounding procedure'
in which the solution to the semidefinite program is converted into a solution
to the $\MAXCUT$ problem. The name `rounding' comes from the fact
that this step can be seen as a way to round the `continuous' SDP solution
into a `discrete' solution to $\MAXCUT$.

Despite intensive research, no better algorithm for $\MAXCUT$ has been found until this day. The best known
NP-hardness result, due to H{\aa}stad, shows that obtaining approximation ratio above $\approx 0.941$ is
NP-hard \cite{Hastad01}. The hardness for approximation factors between $\approx 0.878$ and $\approx 0.941$
was unclear for many years. Recently, it was shown by Khot et al.~\cite{KhotKMO07} that the UGC implies a
tight inapproximability result of $\approx 0.878$, thereby giving a partial answer to this long-standing open
question.

Another problem for which the UGC implies a tight hardness result is
the Vertex Cover problem, where a simple algorithm
gives an approximation factor of $2$ and the UGC implies a hardness
factor of $2-\eps$ for any $\eps>0$~\cite{KhotR03} (whereas the best known
NP-hardness result is $1.36$ \cite{DinurS05}).
The UGC also implies strong inapproximability results for graph coloring
problems~\cite{DinurMR06} and the Sparsest Cut
problem~\cite{KhotV05,ChawlaKKRS06}.


In another line of work, attempts have been made to disprove the conjecture by means of efficient
approximation algorithms for the value of unique games~\cite{Trev05,CMM,GT06,ChlamtacMM06,unique:expander,MakarychevM09}.
So far, however, none of these results disproves this conjecture, and this by itself might be seen
by some as evidence in favor of the conjecture. Among the best algorithms is the one by Charikar et
al.~\cite{CMM} that, given a unique game on alphabet size $k$ whose value is $1-\eps$, outputs a solution of
value $1-O(\sqrt{\eps \log k})$. This does not disprove the conjecture, but instead gives us a lower bound on
$k$ for the conjecture to make sense (see also~\cite{KhotKMO07}). Another recent result is by
Chlamtac et al.~\cite{ChlamtacMM06} who show how to compute, given a unique game with alphabet size $k$ and
$n$ possible questions whose value is $1-\eps$, a solution of value $1-O(\eps\sqrt{\log n \log k})$. This is better than~\cite{CMM} for
small values of $\eps$, but as before, is not enough to disprove the conjecture and only tells us that $k$
and $n$ should be large enough for the conjecture to make sense. Finally, in a recent result,
Makarychev and Makarychev~\cite{MakarychevM09},
improving on earlier work by Arora et al.~\cite{unique:expander}, present an algorithm that given a unique game of value $1-\eps$ for $\eps < c \lambda$, finds a solution of
value at least $1-C \eps/h$ where $\lambda$ is the spectral gap of the graph underlying the unique games instance, $h$
is its edge expansion, and $c,C$ are absolute constants.
This result shows that in order to prove the unique games conjecture, we must
consider graphs with low expansion. We remark that most of these results are based on an SDP relaxation,
followed by a (usually quite sophisticated) rounding procedure.

\paragraph{Games with entangled provers:}
In this paper we consider the model of two-prover one-round games in which the provers are allowed to
share entanglement. (The verifier and all communication remain classical, as before.)
Such games are sometimes known in the quantum information literature as \textit{nonlocal games}
and have their origins in a seminal 1935 paper by Einstein, Podolsky, and Rosen~\cite{EinsteinPR35} and a
1964 paper by Bell~\cite{Bell:64a}. We define the {\em entangled value} of a game as
the maximum success probability achievable by provers that share entanglement.
For instance, it is known that the entangled value of the
CHSH game is $\frac{1}{2}+ \frac{1}{2\sqrt{2}}\approx 85\%$, which is strictly greater than the $75\%$ achievable
without entanglement.
This remarkable ability of entanglement to create correlations that are impossible
to obtain classically (something Einstein referred to as ``spooky")
is one of the most peculiar aspects of quantum mechanics and required many
years to be properly understood.

One motivation for this model comes from the fact that although a verifier can
guarantee that provers don't communicate (by separating them in space, say),
he has no way to guarantee that they don't share entanglement. Therefore, in order
for a proof system with two provers to be sound in our quantum physical world,
we must consider the scenario where the provers share entanglement,
even when the verifier is classical. This is especially true for
multi-prover cryptographic protocols, where the presence of entanglement
could break the security of the protocol.
Another
purely mathematical motivation for this model comes from the hope that
through studying this model, we can reach a better understanding of the
non-classical correlations that arise in quantum mechanics, and even
obtain some new insights into multi-prover games similar to those
obtained from the PCP theorem.

Despite considerable work on this model, our understanding of it is still quite limited. One of the
earliest and most important results in this area is due to Tsirelson~\cite{Tsirelson:80a}, who
showed that, for the special case of unique games with an alphabet of size $k=2$, the entangled
value is given exactly by the optimum of a certain SDP and can therefore be computed efficiently
(see also~\cite{CleveHTW04} where this is made explicit and~\cite{CleveSUU07} for a nice
application of this SDP). Unique games with $k=2$ are also known as XOR-games because one can think
of the two possible answers as a bit, and then the only possible unique constraints are $a \oplus b
= 0$ and $a \oplus b = 1$. This result is in contrast to the (non-entangled) value of an XOR-game,
which is NP-hard to compute exactly or even to approximate (as follows from H{\aa}stad's hardness
result for $\MAXCUT$). Finally, we note that the CHSH game is an XOR-game, and one way to derive
its entangled value of $\frac{1}{2}+\frac{1}{2\sqrt{2}}$ is by computing Tsirelson's SDP.

Unfortunately, our current understanding of the entangled value of games does not extend much
beyond the case of XOR-games. To the best of our knowledge, the only other general result is by
Masanes~\cite{Masanes05}, who shows how to compute the entangled value of games with only two
possible questions to each prover and $k=2$.\footnote{His method also handles the case of more than
two provers.} Although restricted to this very special case, his result still allows us to handle a
few cases not handled by Tsirelson's result (namely, non-unique games for $k=2$ with two
questions).

In all other cases, no method is known to compute or even approximate (with provable guarantees) the
entangled value of a game. Even for some very small fixed size games, there is still uncertainty regarding
their entangled value (see, e.g.,~\cite{Buhrman04}). One recent attempt to handle more general games was made
by Navascues et al.~\cite{NavascuesPA07}, who outlined a hierarchy of SDP relaxations of the entangled value
of a game. Unfortunately, there are no known bounds on the quality of their SDP relaxations.

Another line of work studies the hardness properties of entangled games. Building on work of Kempe
et al.~\cite{kempe:immunization}, Ito, Kobayashi, and Matsumoto~\cite{ItoKM09} have recently shown
that approximating the entangled value of a general game is NP-hard, albeit only to within
exponentially small precision. Strengthening this result to show hardness of approximating the
entangled value to within a constant is still one of the most important open questions in the area.

\paragraph{Our results:}
Our main result is an approximation algorithm for the entangled value of any
unique game. More precisely, our main theorem is the following.

\begin{theorem}\label{thm:mainthmintro}
There exists an efficient algorithm that, given a unique game whose entangled
value is $1-\eps$, outputs a value $\eps/6 \le \eps' \le \eps$ and a description of an entangled strategy for the provers
whose success probability is at least $1-6\eps'$.
\end{theorem}
This theorem gives, for the first time, a way to approximate the entangled value of games with more than two
possible answers. It is also the first provable approximation (as opposed to exact) algorithm for the
entangled value of a game.

Our result shows that the analogue of Conjecture~\ref{conjecture:ugc} for entangled provers
is {\em false}. Indeed, as long as, $6\eps+\delta<1$, our algorithm can efficiently
tell whether the entangled value of a game is at least $1-\eps$ or at most $\delta$.
This can be seen as a (modest) contribution to the understanding of the ever-more-mysterious
unique games conjecture.

It is interesting to compare our algorithm with the approximation algorithms
for the (non-entangled) value of unique games. Given a unique game with entangled value
$1-\eps$, our algorithm outputs an entangled strategy whose value is at least $1-6\eps$.
In contrast, given a unique
game with value $1-\eps$, the algorithms in~\cite{Trev05,CMM,GT06,ChlamtacMM06,unique:expander,MakarychevM09}
output a strategy whose value depends not only on $\eps$ but also on other parameters,
such as the alphabet size $k$ or the expansion of the underlying graph. The fact that our approximation depends only on $\eps$
is crucial.

\paragraph{Techniques:}
The proof of Theorem~\ref{thm:mainthmintro} is based on a semidefinite programming (SDP) relaxation
of the entangled value.
Our SDP turns out to be equivalent to the one used by Khot in~\cite{Khot02} as a relaxation
of the (non-entangled) value of a game (and, in fact, this SDP originates in the work of Feige and Lov\'asz~\cite{FL}).
We note, however, that in the non-entangled case, certain extra constraints are sometimes used that are
not known to hold in the entangled case.

The heart of the proof is in the second step, where we show
how to take a solution to the SDP and transform it into a strategy for entangled
provers. We call this step the `quantum rounding' step in analogy with the
rounding procedure used in the non-entangled case.
We hope that this novel technique will be useful for other problems
as well. The main idea in our rounding step is to use the vectors given
by the SDP solution as a quantum measurement performed by the
provers on a maximally entangled state shared by them.

\paragraph{Extensions:}
We present two extensions of our main theorem. The first involves a special case of unique games which we
call {\em uniform} unique games. These are unique games for which there exists an optimal strategy in which
each prover's answer distribution is uniform (for each question). As we show later, any unique game in which
the verifier's decision is based solely on $a-b \pmod k$ is a uniform unique game. This includes XOR-games as
well as the unique games constructed in~\cite{KhotKMO07}. For this special case, we show that the factor $6$
in our main theorem can be improved to $4$. We also extend our main theorem to $d$-to-$d$ games, which are
another type of game considered in~\cite{Khot02}. Namely, we show that Khot's conjecture for $d$-to-$d$ games
is false in the case that the provers share entanglement.

\paragraph{Parallel repetition:}
Our semidefinite programming relaxation also allows us to show a parallel repetition theorem for unique
entangled games. Parallel repetition for non-entangled classical games has been investigated extensively,
with early work culminating in Raz's parallel repetition theorem \cite{Raz98}.
In the case of entangled games, no parallel repetition theorem is known, and proving one
seems even more challenging than in the case of non-entangled games. The only special case where
parallel repetition is known to hold is for entangled XOR-games~\cite{CleveSUU07}. We
show that this result can be extended to unique games, albeit with somewhat weaker quantitative behavior (see
Section \ref{sec:rep} for a precise statement).

Our approach to prove parallel repetition is similar to the one taken by Cleve et al. (which in fact dates
back to earlier work by Feige and Lov\'asz \cite{FL}): we show that a certain bipartite SDP relaxation of the
entangled game is multiplicative. The latter fact essentially follows from a recent result of Mittal and
Szegedy \cite{MittalSzegedy:SDPtensor}. See Section \ref{sec:rep} for details.

\paragraph{Discussion:}
Our work gives for the first time a way to approximate the entangled value of games with more than
two possible answers. One open question this raises is whether there exist better algorithms for
approximating (or even computing exactly) the entangled value of a unique game. So far we only know
of such a result in the case $k=2$, where Tsirelson's SDP gives an exact
answer~\cite{Tsirelson:80a}. Extending this to $k>2$ might require improving our quantum rounding
procedure, and might also involve the use of a tighter SDP, perhaps taken from  the SDP hierarchy
outlined in~\cite{NavascuesPA07}. Another open question is whether our quantum rounding technique
can be used for other types of games. One good candidate are games with inequality constraints,
such as $\MAXKCUT$, as those are relatively well-understood~\cite{KhotKMO07}. One might also hope
to extend our results to the case of general entangled games.

\paragraph{Strong violation of Bell inequalities:}
Games exhibiting a gap between their entangled value and non-entangled value
are of great interest to physicists, for possible use in experiments whose
goal is to demonstrate the presence of quantum entanglement
(see, e.g.,~\cite{werner01:_bell} and references therein).
Such games are said in the physics literature to exhibit a `violation of Bell inequalities'.
By combining our
main result with a remarkable construction by Khot and Vishnoi~\cite{KhotV05},
we can obtain unique games whose entangled value is very close to $1$ even
though their value is very close to $0$. Previously, such large gaps
were known only for non-unique games (such as the parallel repetition of the Magic Square game).
The simpler structure of unique games might be an advantage in certain circumstances.
A related result was recently established for three-prover games with binary
answers~\cite{perez-garcia:_unboun}.

In more detail, Khot and Vishnoi constructed for any $k \ge 1$ and $\eta > 0$,
a unique game with $2^k/k$ questions to each prover
and answer alphabet of size $k$ for which the value of our SDP relaxation is at least $1-9 \eta$ and whose
(non-entangled) value is at most $2/k^{\eta}$.\footnote{Strictly speaking, their construction
gives a general constraint graph, and not a two-prover game as needed in our case.
In order to derive a two-prover game from their construction, simply choose a random constraint
and then randomly send one question to each prover.}
(We note that the existence of unique games whose SDP value is close to $1$ and whose
value is close to $0$ follows from the UGC, but Khot and Vishnoi's result
is unconditional and also gives explicit parameters.)
By combining their result with Theorem~\ref{thm:mainthmintro}, we obtain
that for any $k \ge 1$ and $\eta > 0$,
there exists a unique game $G$ with $2^k/k$ questions to each prover
and answer alphabet of size $k$ for which the entangled value is at least $1-54 \eta$ and whose
(non-entangled) value is at most $2/k^{\eta}$.

\section{Preliminaries}

We study \textit{one-round two-prover cooperative games of
incomplete information}, also known in the quantum information literature as
\textit{nonlocal games}.  In such a game, a referee (also called the verifier) asks questions to two provers, Alice and Bob,
who cooperate with each other. A  game $G = G(\pi, V)$ is specified by a set $Q$ and a number $k \ge 1$, a probability
distribution $\pi:Q\times Q\to [0,1]$, and a predicate $V:[k] \times [k] \times Q \times Q \to \{0,1\}$.  The
game proceeds as follows: the referee samples $(s, t) \in Q \times Q$ according to $\pi$ and sends question
$s$ to Alice and question $t$ to Bob. Alice replies with an answer $a
\in [k]$, and Bob with an answer $b \in [k]$.  The provers win if and only if $V(a,b\,|\,s,t)=1$.\footnote{We
write $V(a,b\,|\,s,t)$ for $V(a,b,s,t)$ to distinguish variables for questions
  from variables for answers.}
The provers are allowed to agree on a strategy before the game starts, but are not allowed to communicate
with each other after receiving their questions. The {\em value} of a game is the maximum probability with
which the provers can win. The provers may share randomness, but it is easy to see that this does not
increase the value of the game.

The provers can also share an entangled state, which can sometimes increase their winning probability
(for background on quantum information see, e.g., \cite{NC}). We
therefore define the {\em entangled value} of a game to be the highest winning probability of entangled
provers. Let us define this more explicitly.
In general, a strategy for entangled provers is described by a shared (possibly mixed) quantum state,
as well as a general measurement on Alice's part of the state for each of her questions,
and a general measurement on Bob's part of the state for each of his questions. On obtaining
question $s$, Alice performs the measurement corresponding to $s$ on her part of the state
and returns as answer the result; Bob's behavior is similar.
By standard arguments, we can assume without loss of generality that
Alice and Bob share a {\em pure} quantum state $\ket{\psi} \in \CC^{d \times d}$
for some $d \ge 1$, and that, moreover, they use {\em projective measurements}, i.e.,
for each $s$ Alice's measurement is described by $\{A_a^s\}_a$ where
the $A_a^s$ are orthogonal projectors and $\sum_a A_a^s=\id$, and similarly Bob uses
measurements $\{B_b^t\}_b$. By definition, the probability that on questions $s,t$
Alice answers $a$ and Bob answers $b$ is given by $\bra \psi
A_a^s \otimes B_b^t \ket \psi$. Therefore, the entangled value of $G$ can be written as
\begin{align*}
  \omega^* (G) = \lim_{d \to \infty} \max_{\ket \psi\in \CC^d \otimes \CC^d} \max_{A_a^s, B_b^t}\sum_{abst}
  \pi(s,t) V(a,b\,|\,s,t) \bra \psi A_a^s \otimes B_b^t \ket \psi.
\end{align*}

We shall be concerned with games of a specific form.
\begin{definition}
  \label{definition:1}
A game is termed \emph{unique} if we can associate a permutation $\sigma_{st}$ on $[k]$ with each pair
of questions $(s,t)$ such that $V(a,b\,|\,s,t) = 1$ if and only if $b = \sigma_{st}(a)$.
\end{definition}

Next we define \emph{linear games}, an important special case of unique games
that has been extensively studied in the literature (see, e.g., \cite{Hastad01,KhotKMO07}).
Linear games are a natural generalization of XOR-games to larger
alphabet size, and are defined as follows.
\begin{definition}
  \label{def:linear}
A game is termed \emph{linear} if there
is a way to identify $[k]$ with some Abelian group $H$ of size $k$ and
 a function $W:Q \times Q \to H$ such that $V(a,b\,|\,s,t)=1$ if and only if  $a - b = W(s,t)$ in $H$.
\end{definition}
It is easy to see that any linear game is in particular a unique game.
We also consider games in which the answers of each prover in an optimal
strategy are distributed uniformly in $[k]$.
\begin{definition}
  \label{definition:3}
A game is termed \emph{uniform} if there exists an optimal strategy for entangled provers
in which, for each prover and for each question, the marginal distribution of his answers
is uniform over $[k]$.
\end{definition}
One easy observation is that any linear game is also uniform. To see this, notice
that Alice and Bob, by using their shared randomness (or entanglement), can choose an element $c$ uniformly
from $H$ and add it to their responses. This does not change their probability of winning the game, but
causes each party's output to be uniformly distributed in $H$.

Finally, we consider more general games known as {\em $d$-to-$d'$ games}.
\begin{definition}\label{definition:dd}
A game has the \emph{$d$-to-$d'$ property} if for each pair of questions $(s,t)$
and each answer $a$ of the first prover, there are at most $d$ answers $b$
of the second prover for which $V(a,b\,|\,s,t) = 1$, and similarly, for each
answer $b$ of the second prover, there are at most $d'$ answers $a$ for which $V(a,b\,|\,s,t) = 1$.
\end{definition}
In~\cite{Khot02}, Khot conjectured that for any $\delta > 0$ there exists a $k = k(\delta)$ such
that it is \NP-hard to determine whether, given a $2$-to-$1$ game with answers from a domain of size $k$,
its value is exactly $1$ (i.e., perfectly satisfiable) or at most $\delta$. This conjecture was
used in proving the hardness of graph coloring problems~\cite{DinurMR06}. In Theorem~\ref{thm:ddgames}
below we show that if the provers are allowed entanglement, then this conjecture, as well as its
extension to $d$-to-$d'$ games, is false.

\section{SDP Relaxation}\label{sec:relaxation}

We use the following SDP relaxation for the entangled value of an arbitrary two-prover one-round game. The
SDP maximizes over the real vectors $\{u_a^s\}$, $\{v_b^t\}$, and $z$.

\begin{program}
\caption{}\label{sdp:1}
\begin{tabular}{p{0.12 \textwidth}p{0.8 \textwidth}}
\textbf{Maximize:} & $\sum_{abst} \pi(s,t) V(a,b\,|\,s,t) \ip{u_a^s, v_b^t}$\\
\textbf{Subject to:} & $\|z\|=1$\\
& $\forall s,t,~\sum_a u_a^s = \sum_b v_b^t = z$\\
& $\forall s,t,~\forall a \neq b,~ \ip{u_a^s,u_b^s}=0$ and $\ip{v_a^t,v_b^t}=0$\\
& $\forall s,t,a,b,~\ip{u_a^s,v_b^t} \ge 0$\\
\end{tabular}
\end{program}
\begin{remark}\label{remark:1}
Note that the second constraint is, strictly speaking, not an SDP constraint. However, it is
easy to see that there is an equivalent formulation in SDP language. For instance, we can replace the first
two constraints by $\sum_{a,b} \ip{u_a^s,v_b^t}=1$, $\sum_a \ip{u_a^s,u_a^s}=1$ and $\sum_b
\ip{v_b^t,v_b^t}=1$.
\end{remark}

 For a game $G$, let $\omega_{\sdp1}(G)$ be
the value of SDP~\ref{sdp:1}. We start by showing that it is indeed a relaxation of the
entangled value of the game.

\begin{lemma}
  \label{lemma:2} Let $G = G(\pi, V)$ be a (not necessarily unique) one-round two-prover game. Then
  $\omega^*(G) \leq \omega_{\sdp1}(G).$
\end{lemma}
\begin{proof}
Consider any strategy for the entangled provers, specified by a state $\ket \psi \in \CC^{d \times d}$ and
projectors $\{A_a^s\}$ and $\{B_b^t\}$. Define the vectors $\tilde{u}_a^s = (A_a^s \otimes \id) \ket \psi$
and ${\tilde v}_b^t = (\id \otimes B_b^t) \ket \psi$ in $\CC^{d \times d}$. Consider now the {\em real}
$2d^2$-dimensional vectors defined by $u_a^s={\rm Re}(\tilde{u}_a^s) \oplus {\rm Im}(\tilde{u}_a^s)$, $v_b^t={\rm Re}
(\tilde{v}_b^t) \oplus {\rm Im}(\tilde{v}_b^t)$ and $z={\rm Re}(\ket{\psi})\oplus {\rm Im}(\ket{\psi})$. Note that because
$\ip{\tilde{u}_a^s,\tilde{v}_b^t} =\bra \psi A_a^s \otimes B_b^t \ket \psi$ is real, we have that
$$\ipb{u_a^s,v_b^t}={\rm Re}(\tilde{u}_a^s){\rm Re}(\tilde{v}_b^t)+{\rm Im}(\tilde{u}_a^s){\rm Im}(\tilde{v}_b^t)=
 {\rm Re}(\ipb{\tilde{u}_a^s,\tilde{v}_b^t})=
\bra \psi A_a^s \otimes B_b^t \ket \psi \geq 0.$$ The other constraints follow from the observations that
$\sum_a \tilde{u}_a^s = \ket{\psi}=\sum_b {\tilde v}_b^t$, that $\ip{z,z} = {\rm Re}(\langle \psi | \psi \rangle) = 1$,
and that for $a \neq b$,
 $$\ipb{u_a^s,u_b^s}={\rm Re}(\ipb{\tilde{u}_a^s,\tilde{u}_b^s})={\rm Re}(\bra \psi A_a^s A_b^s \otimes I \ket \psi)=0,$$
 since $A_a^sA_b^s=0$, and similarly $\ip{v_a^t,v_b^t}=0$.
\end{proof}

In the case of uniform games there exists an optimal strategy in which
the provers' output distribution is uniform on $[k]$. This allows us to add the following constraint to the
original SDP:

\begin{quote}
{\bf Additional constraint for SDP 2:} $\quad \forall s,t,a,b,~\|u_a^s\| = \|v_b^t\| =
1/\sqrt{k}$.
\end{quote}
This gives a more constrained SDP relaxation for uniform games, which we call SDP 2 and whose
value we denote by $\omega_{\sdp2}(G)$. To see that
this is indeed a relaxation of uniform games, note that with the notation of the
proof of Lemma \ref{lemma:2},
$$\ip{u_a^s,u_a^s} =
\ip{\tilde{u}_a^s,\tilde{u}_a^s} = \bra \psi A_a^s \otimes \id \ket \psi, $$
which is equal to $\frac{1}{k}$ since the last expression is exactly Alice's
marginal distribution, and similarly for $v_b^t$. This extra constraint
will allow us to slightly improve our quantum rounding procedure.

\section{Quantum Rounding}

In this section we describe how to round the solution of our SDP to a quantum strategy. We start with an
informal outline of the rounding algorithm for the special case of {\em uniform} unique games, and then describe how to
modify it for general unique games. The formal description of the rounding algorithm is given as
Algorithm~\ref{alg:1} below, and it uses a particular measurement, given here as Measurement~\ref{meas:1}, as
a subroutine.

Our goal is the following. The SDP relaxation of a game gives us a solution $\{u_a^s\},\{v_b^t\}$, where for
fixed $s,t$ the inner products $\ipb{u_a^s,v_b^t}$ can be interpreted as a joint probability distribution on $(a,b)$ (note
that they are non-negative and sum to $1$). The marginal distribution on $a$ is given by $\|u_a^s\|^2$
(since $\sum_b \ipb{u_a^s,v_b^t} = \ipb{u_a^s,z} = \sum_{a'} \ipb{u_a^s,u_{a'}^s} = \|u_a^s\|^2$)
and on
$b$ by $\|v_b^t\|^2$; in particular, for SDP 2 these marginal distributions are {\em uniform}. The value of the SDP then
represents the winning probability in the corresponding game (given by $\pi$ and $V$), when the provers
answer according to this probability distribution. Hence we would like to design an entangled strategy that {\em
reproduces} this probability distribution as closely as possible, so that its winning probability will be
close to the value of the SDP solution.

The basic idea is to use the solution to the SDP to define a measurement for Alice and Bob on the {\em
maximally entangled} state $\ket{\psi} = \frac{1}{\sqrt{n}}\sum_{i=1}^n \ket{i,i}$. This state has the
property that for any orthonormal basis of {\em real} vectors $\{\ket{u_i}\}_{i=1}^n$ it can be written as
$\ket \psi=\frac{1}{\sqrt{n}}\sum_{i=1}^n \ket{u_i,u_i}$. This implies that if Alice measures in such a basis
and obtains outcome $i$, then Bob's state collapses to $\ket{u_i}$. If Bob then measures in a basis $\{\ket
{v_j}\}$ in which one of the vectors, say $\ket{v_j}$, is close to $\ket{u_i}$, then Bob's measurement
outcome is likely to be $j$.

We now describe our rounding algorithm for the simpler case of uniform unique games, using SDP 2.
Consider a solution of SDP 2 and assume that it lies in $\RR^n$ for some $n\geq 1$. Assume moreover that the
value of this solution is $1-\eps$ for some small $\eps >0$. This means that for a typical pair $s,t$ the sum
$\sum_{a=1}^k \ipb{u_a^s,v_{\sigma_{st}(a)}^t}$ is close to $1$, and hence, since the norms of all these vectors are
$1/\sqrt{k}$,  $u_a^s$ is typically close to $v_{\sigma_{st}(a)}^t$. We now use this solution to define local
projective measurements on the $n$-dimensional maximally-entangled state $\ket{\psi} =
\frac{1}{\sqrt{n}}\sum_{i=1}^n \ket{i,i}$. For fixed $s$, the $k$ vectors $u_a^s$ are orthogonal, and so,
after normalization, define part of a basis. We complete this to a basis of $\RR^n$  in an arbitrary way.
When Alice is asked question $s$, she measures her half of $\ket{\psi}$ in this basis, outputting $a$ if her
measurement result corresponds to the basis element $u_a^s$, and outputting (for the moment) nothing if she
obtains one of the extra basis elements. Similarly, when Bob is asked question $t$, he measures his half of
$\ket {\psi}$ in a basis that contains the vectors $\{v_b^t\}_b$, outputting $b$ if his measurement result
corresponds to the vector $v_b^t$, and nothing otherwise.

This rounding scheme has two properties --- one good, the other bad. First the good property: if Alice outputs
$a$,  then Bob will output $\sigma_{st}(a)$ with high probability, since the vector $u_a^s$ is close to the vector
$v_{\sigma_{st}(a)}^t$. Hence we have that if Alice does give an answer, Bob's answer will be correct with high
probability. The problem is that the probability that Alice does give an answer is $k/n$, which is typically
very small.

Luckily, the solution to this problem is easy. If Alice doesn't obtain a suitable outcome, she just starts
over, performing her measurement on a fresh maximally entangled state. She keeps performing her measurement
on fresh states until she obtains an outcome corresponding to a vector $u_a^s$. Bob does likewise, performing
his measurement on fresh states, making sure to use them in the same order as Alice. This fixes the problem
above, since now Alice answers with probability $1$. In doing so, however, it seems we have created a new
problem: it is entirely possible that Alice and Bob's measurements will succeed on different copies of the
maximally-entangled state, in which case their answers won't be correlated at all. But here we are saved by
the good property of our measurement: if Alice does obtain an outcome $a$ in a given round, then Bob will
obtain a correct outcome $b$ with high probability, assuming he hasn't already answered. Conversely, if Bob
obtains an outcome $b$ in a given round, then Alice will obtain a correct outcome $a$ with high probability,
assuming she hasn't already answered. All this means that, with high probability, their measurements succeed
on the same copy of $\ket {\psi}$.

Let us now consider general (not necessarily uniform) unique games. Our starting point now is a solution to
SDP~\ref{sdp:1} of value $1-\eps$. Again, for a typical pair $s,t$ we have that
$\sum_{a}\ipb{u_a^s,v_{\sigma_{st}(a)}^t}$ is close to $1$. However, in our rounding algorithm we have to add
an extra step to account for the fact that the vectors $u_a^s, v_b^t$ might not all be of the same length.
The projective measurement that we have constructed does not take this into account, because all basis
vectors after renormalization are equally likely to occur. Recall that our goal is to reproduce the
joint probability distribution on $a,b$ given by $\ipb{u_a^s,v_b^t}$, and, in particular, its marginal distributions on $a$ and $b$,
given by $\|u_a^s\|^2$ and $\|v_b^t\|^2$: Our rounding algorithm has to ensure that outcomes that correspond
to short vectors $u_a^s$ or $v_b^t$ should be less likely than those corresponding to longer vectors. To this
end we use a rejection sampling technique, as follows. Alice samples $\lambda$ uniformly from $[0,1]$. If Alice's measurement outcome is $a$, she outputs it if and only if $\lambda \leq
\|u_a^s\|^2$.  This ensures that the probability that Alice answers $a$ is the same as that given by the SDP
relaxation. Similarly, if Bob's measurement outcome is $b$, he
outputs it if and only if $\lambda \leq \|v_b^t\|^2$. Again a problem arises that, even if Alice's and Bob's
outcomes are otherwise correlated, the rejection sampling could make Alice accept and Bob reject (or vice versa) on the same
copy of $\ket \psi$. Luckily, we are helped again by the fact that on average, $u_a^s$ and
$v_{\sigma_{s,t}(a)}^t$ are close, and in particular have comparable length. Therefore, to coordinate the rejection
sampling procedure, Alice and Bob will use a {\em shared} random variable $\lambda$ in each step, which means
that with high probability they will either both accept or both reject.
(Note that it is easy to obtain a shared random variable from shared entanglement.)

\begin{algorithm}
\caption{Quantum rounding for unique games.}\label{alg:1}
\begin{tabular}{p{0.14 \textwidth}p{0.8 \textwidth}}
\textbf{Setup:} & Alice and Bob share many copies of an $n$-dimensional maximally entangled
state $\ket \psi = \frac{1}{\sqrt{n}} \sum_{i=1}^n \ket{i,i}$, for some fixed basis $\{\ket{i}\}$ of $\CC^n$,
as well as a sequence $\Lambda = (\lambda_1, \lambda_2, \dots)$ of real numbers, where the $\lambda_i$ are independent and each is sampled uniformly from $[0,1]$. \\
\textbf{Alice:} & On input $s$, performs the measurement \textsc{Measure}($u_1^s,u_2^s, \ldots, u_k^s$)
 on her share of the maximally entangled states and the sequence $\Lambda$. \\
\textbf{Bob:} & On input $t$, performs the measurement \textsc{Measure}($v_1^t,v_2^t, \ldots,
v_k^t$)
 on his share of the maximally entangled states and the sequence $\Lambda$. \\
\end{tabular}
\end{algorithm}

\begin{measurement}
\caption{The measurement \textsc{Measure}$({x_1, x_2, \ldots, x_k})$ used in
Algorithm~\ref{alg:1}.} \label{meas:1}
\begin{tabular}{p{0.24 \textwidth}p{0.7 \textwidth}}
\textbf{Input:} & A state on a Hilbert space $\calH = \bigotimes_{r=1}^\infty \calH_r$, where each $\calH_r\cong \CC^n$,
 and a sequence of real numbers $\Lambda = ( \lambda_1, \lambda_2, \ldots)$, where each $\lambda_r \in [0,1]$.\\
\textbf{Parameters:} & $k$ orthogonal vectors $x_1,x_2, \ldots, x_k \in \RR^n$. \\
\textbf{Output:} & An integer $m \in \{1, 2, \ldots, k\}$. \\
\textbf{Measurement:}& Define a POVM on $\CC^n$ with elements
\begin{align*}
P_i = \kbbig{\frac{x_i}{\|x_i\|}} \text{ for $i = 1, 2, \ldots, k$ and } P_0 = \id - \sum_{i=1}^k P_i,
\end{align*}
where for a vector $w \in \RR^n$ we
write $\ket{w}=\sum_i (w)_i \ket{i}$ for its embedding into $\CC^n$.\\
& \textbf{For} $r=1,2,\ldots$ do: \newline \mbox{\hspace{1cm}} Measure $\calH_r$ using POVM $(P_0, \ldots,
P_k)$, obtaining outcome $m$. \newline \mbox{\hspace{1cm}} \textbf{If} ($m\neq 0$ and $\lambda_r \leq
\|x_m\|^2$) \textbf{then} output $m$ and exit.
\end{tabular}
\end{measurement}

\subsection{Analysis of the measurement procedure}

\begin{lemma}
  \label{lemma:1}
Let $x_1,\ldots,x_k$ and $y_1,\ldots,y_k$ be two sequences of orthogonal vectors in $\RR^n$ such that
$\sum_{i=1}^k \|x_i\|^2=\sum_{i=1}^k \|y_i\|^2=1$. Assume Alice and Bob apply Measurement~\ref{meas:1}, Alice
using $(x_i)$ and Bob using $(y_i)$. For any $i,j \in \{1,\ldots,k\}$ define
$$ q_{i,j} := \ipB{\frac{x_i}{\|x_i\|},\frac{y_j}{\|y_j\|}}^2 \min(\|x_i\|^2, \|y_j\|^2)$$
and let $q_{\rm total} := \sum_{i,j} q_{i,j}$.
Then for any $i,j\in \{1,\ldots,k\}$, the probability that Alice outputs $i$
and Bob outputs $j$ is at least
$$ \frac{q_{i,j}}{2-q_{\rm total}} .$$
\end{lemma}

\begin{proof}
We start by analyzing one round of the measurement, i.e., Alice and Bob share a maximally entangled
$n$-dimensional state and a random number $\lambda \in [0,1]$. Each performs a measurement given by his or
her input vectors, and outputs the outcome $m$ if $m\neq 0$ and $\lambda \le \|x_m\|^2$ (resp. $\lambda \le
\|y_m\|^2$), or nothing otherwise.

A round can end in one of four possible ways, to which we assign probabilities as follows:
\begin{itemize}
\item ($p_\text{done}$) Both Alice and Bob give an output;
\item ($p_{\text{fail},1}$) Alice gives an output while Bob does not;
\item ($p_{\text{fail},2}$) Bob gives an output while Alice does not;
\item ($p_\text{retry}$) Neither Alice nor Bob gives an output.
\end{itemize}
Hence $p_\text{done}+p_{\text{fail},1}+p_{\text{fail},2}+p_\text{retry}=1$. Let us also define $p_{i,j}$ for
$i,j \in \{1,\ldots,k\}$ as the probability that Alice outputs $i$ and Bob outputs $j$ in one round. Notice
that $p_\text{done} = \sum_{i,j=1}^k p_{i,j}.$

We now compute each of these probabilities. By construction, the probability that Alice obtains an outcome $i
\neq 0$ from her POVM is exactly $1/n$. Conditioned on that happening, Bob's state collapses to the pure
state given by the vector $\ket{x_i}/\|x_i\|$. Therefore, the conditional probability that he obtains an outcome $j\neq
0$ in his POVM is given by $\ipb{\frac{x_i}{\|x_i\|},\frac{y_j}{\|y_j\|}}^2$. Finally, conditioned on Alice
measuring $i \neq 0$ and Bob measuring $j\neq 0$, the probability that both actually output their values is
$\min(\|x_i\|^2, \|y_j\|^2)$. Hence we see that for any $i,j \in \{1,\ldots,k\}$, the probability that in one
round of the measurement Alice outputs $i$ and Bob outputs $j$ is
\begin{align*}
  p_{i,j} := \frac{1}{n}\ipB{\frac{x_i}{\|x_i\|},\frac{y_j}{\|y_j\|}}^2 \min(\|x_i\|^2, \|y_j\|^2)=\frac{1}{n}q_{i,j}.
\end{align*}
Moreover, it is easy to see that the probability that Alice gives an output is
$$ \sum_{i=1}^m \frac{1}{n} \|x_i\|^2 = \frac{1}{n}$$
and similarly for Bob. This implies that
$$ p_{\text{fail},1} = p_{\text{fail},2} = \frac{1}{n} - p_\text{done}.$$

To complete the proof, let us consider the probability that in Measurement~\ref{meas:1},
Alice outputs $i$ and Bob outputs $j$. This probability is lower bounded by the probability
that Alice outputs $i$ and Bob outputs $j$ in the same round. The latter probability is given by
\begin{align*}
  \sum_{r=0}^\infty (p_\text{retry})^r p_{i,j} =
  \frac{p_{i,j}}{1-p_\text{retry}} =
  \frac{p_{i,j}}{p_\text{done}+ p_\text{fail,1}+ p_\text{fail,2}} =
  \frac{p_{i,j}}{\frac{2}{n} - p_\text{done}} =
  \frac{q_{i,j}}{2-q_\text{total}} .
\end{align*}
\end{proof}

\begin{corollary}
  \label{cor:goodanswers}
Let $V$ be a subset of $\{1,\ldots,k\}^2$. Then, in the setting of Lemma~\ref{lemma:1},
the probability that Alice's output $i$ and Bob's output $j$ are such that $(i,j) \in V$
is at least
$$ \frac{p_V}{2-p_V} \ge 1- 2(1-p_V), $$
where
$$ p_V := \sum_{i,j \in V} \ipB{\frac{x_i}{\|x_i\|},\frac{y_j}{\|y_j\|}}^2 \min(\|x_i\|^2, \|y_j\|^2).$$
\end{corollary}

\subsection{Analysis of the quantum rounding}

We first analyze the easier case of uniform unique games.
We remark that we could have slightly simplified the algorithm for this case by
avoiding the rejection sampling step, but for convenience we keep it since
it does not affect our results.

\begin{theorem}[Uniform unique games]
  \label{thm:lingames}
Let $G$ be a uniform unique game. Suppose that $\omega_{\sdp2}(G) = 1 -\eps$. Then $\omega^*(G) \geq 1 - 4 \eps$.
\end{theorem}
\begin{proof}
Fix a solution $\{u_a^s\}$, $\{v_b^t\}$, $z$ to SDP~2 with value $1-\eps$ and consider the
strategy of Alice and Bob given by Algorithm~\ref{alg:1}. Our goal is to show that this strategy has success
probability at least $1-4\eps$. In order to show this, it suffices to show that for any questions $s,t$, the
success probability of Alice and Bob on these questions is at least $1 - 4(1-\sum_{ab} V(a,b\,|\,s,t)
\ip{u_a^s, v_b^t})$.

So from now on fix a pair of questions $s,t$ and let $\sigma$ be the permutation corresponding to the
constraint between $s$ and $t$, i.e., $V(a,b\,|\,s,t)=1$ if and only if $b = \sigma(a)$. For $i=1,\ldots,k$,
define $u_i = u_i^s$ and $v_i = v_{\sigma(i)}^t$. Suppose that $\sum_i \ip{u_i,v_i} \ge 1-\tilde{\eps}$ for
some $\tilde{\eps} \ge 0$ and recall that our goal is to show that Alice and Bob succeed with probability
at least $1-4{\tilde \eps}$. By Corollary~\ref{cor:goodanswers}, their success probability is at least
$$ p_{\rm succ} \ge 1 - 2(1-p'_{\rm succ}),$$
where
\begin{align}\label{eq:pprimesucc}
p'_{\rm succ} = \sum_{i=1}^k \ipB{\frac{u_i}{\|u_i\|},\frac{v_i}{\|v_i\|}}^2 \min(\|u_i\|^2, \|v_i\|^2).
\end{align}
It therefore suffices to show that $p'_{\rm succ} \ge 1-2\tilde \eps$.
Using the extra constraints in SDP~2 and the Cauchy-Schwarz inequality,
$$ p'_{\rm succ} = k \ \sum_{i=1}^k \ip{u_i,v_i}^2 \ge \bigg( \sum_{i=1}^k \ip{u_i,v_i}\bigg)^2 \ge 1-2\tilde{\eps}.$$
\end{proof}

\begin{remark}\label{remark:2}
Notice that among the five constraints in SDP $2$ we only used the third constraint in
SDP~\ref{sdp:1} on the orthogonality of the vectors, and the additional constraint of SDP 2. Moreover the
vector $z$ is unnecessary.
\end{remark}

\begin{theorem}[Unique games]
  \label{thm:uniqgames}
Let $G$ be a unique game. Suppose that $\omega_{\sdp1}(G) = 1 -\eps$. Then $\omega^*(G) \geq 1 - 6 \eps$.
\end{theorem}
\begin{proof}
As in the proof of Theorem~\ref{thm:lingames}, we have vectors $u_i$, $v_i$, this time coming from SDP~\ref{sdp:1},
satisfying $\sum_i \ip{u_i,v_i} \ge 1-\tilde{\eps}$.
Our goal now is to show that $p'_{\rm succ} \ge 1-3\tilde{\eps}$ where $p'_{\rm succ}$ is defined as in Eq.~\eqref{eq:pprimesucc}.

Let $F := \sum_i \|u_i\|  \|v_i\|$. We
first notice that
$$ F
   \le \bigg(\sum_i \|u_i\|^2\bigg)^{1/2} \bigg(\sum_i \|v_i\|^2\bigg)^{1/2} = 1.$$
Define
$$ p''_{\rm succ} := \sum_i \ipB{\frac{u_i}{\|u_i\|},\frac{v_i}{\|v_i\|}}^2 \|u_i\| \|v_i\|.$$
Then, by convexity,
\begin{align*}
p''_{\rm succ} &= F \ \sum_i \frac{\|u_i\| \|v_i\|}{F} \ipB{\frac{u_i}{\|u_i\|},\frac{v_i}{\|v_i\|}}^2 \\
 & \ge F \bigg( \sum_i \frac{\|u_i\| \|v_i\|}{F} \ipB{\frac{u_i}{\|u_i\|},\frac{v_i}{\|v_i\|}} \bigg)^2 \\
 & = \frac{1}{F} \bigg( \sum_i \ip{u_i,v_i} \bigg)^2 \\
 & \ge 1 - 2\tilde \eps.
\end{align*}
Moreover, using the fact that for any nonnegative $a,b \in \R$, $\sqrt{ab} - \min(a,b) \le |a-b|/2$,
\begin{align*}
p''_{\rm succ} - p'_{\rm succ}&\le
  \frac{1}{2} \sum_i \ipB{\frac{u_i}{\|u_i\|},\frac{v_i}{\|v_i\|}}^2 | \|u_i\|^2 -  \|v_i\|^2 | \\
  &\le   \frac{1}{2} \sum_i | \|u_i\|^2 -  \|v_i\|^2 | \\
  &=  \frac{1}{2} \sum_i | \ip{u_i,z} -  \ip{v_i,z} | \\
  &=  \frac{1}{2} \sum_i \bigg| \sum_{j \neq i} \ip{u_i,v_j} -  \sum_{j \neq i} \ip{v_i,u_j}  \bigg| \\
  & \le \frac{1}{2} \bigg( \sum_i \sum_{j \neq i} \ip{u_i,v_j} + \sum_i  \sum_{j \neq i} \ip{v_i,u_j} \bigg)
  \le \tilde \eps.
\end{align*}
\end{proof}

\begin{remark}\label{remark:3}
The above analysis is `locally tight' in the following sense.
For any small $\tilde{\eps}>0$ and large enough $k$, there exist two sequences of orthogonal vectors
$u_1,\ldots,u_k$ and $v_1,\ldots,v_k$ satisfying
(i) $\sum_i u_i = \sum_i v_i$ has norm $1$,
(ii) for all $i,j$, $\ip{u_i, v_j} \ge 0$, (iii)
$\sum_{i=1}^k \ip{u_i,v_i} = 1-\tilde{\eps}$, and (iv) the probability that the quantum
rounding procedure produces a pair $(i,j)$ with $i=j$ is roughly $1-6\tilde{\eps}$.
Let $a = \sqrt{(1-\tilde{\eps})/k}$ and $b = \sqrt{2\tilde{\eps}/k}$, and let $e_1,\ldots,e_k,f_1,\ldots,f_{k/2}$
be orthonormal unit vectors. Our vectors are given by $u_i = a e_i + b f_i$, $v_i =a e_i$ for $i=1,\ldots,\frac{k}{2}$,
and $u_i=a e_i$, $v_i = a e_i + b f_{i-\frac{k}{2}}$ for $i=\frac{k}{2}+1,\ldots,k$.
\end{remark}
\medskip

Our final theorem deals with $d$-to-$d$ games, and uses the following combinatorial claim.

\begin{claim}\label{clm:acycligraph}
Let $(V,E)$ be a directed acyclic graph with non-negative weights associated with its vertices,
and let $V' \subseteq V$ denote the set of vertices with outdegree zero.
Assume, moreover, that all indegrees are at most $D$, and that the weight of each
vertex in $V\setminus V'$ is smaller by a factor of at least $2D$ than the
sum of weights of its out-neighbors. Then the total weight in $V'$ is at least half
the total weight in $V$.
\end{claim}

For example, for any $m,D\ge 1$, consider the graph with $mD$ nodes, arranged in $m$
layers of $D$ nodes each. We assign weight $2^i$ to nodes in layer $i$,
and place a directed edge from any node in layer $i$ to any node in layer $i+1$.
Notice that this graph satisfies the conditions of the claim, and
that the weight of the nodes with outdegree zero, namely the nodes in layer $m$,
is essentially half of the total weight.

\begin{proof}
Assume without loss of generality that $V=\{1,\ldots,n\}$, that all edges are
facing forward (i.e., are of the form $(i,j)$ with $j>i$), and that $V'=\{n-m+1,\ldots,n\}$
for $m =|V'| \ge 1$. Let $w_i$ denote the weight of vertex $i$. Consider the following process. Initially, for $i=1,\ldots,n$,
set $a_i$ to be $w_i$.
Then, for $i=1,\ldots,n-m$ do the following: for each edge $(i,j)$ leaving $i$,
add $w_j/D$ to $a_j$, and then set $a_i = 0$.

Notice the following two properties. First, since the in-degrees
are at most $D$, we always have $a_i \le 2 w_i$ for all $i$. Second, $\sum_{i=1}^n a_i$
never decreases during this process. The reason is that although we decrease the sum by $a_i$
when we set $a_i$ to $0$, we also increase it by the sum of weights of $i$'s out-neighbors divided by $D$,
which is by assumption at least $2w_i \ge a_i$. Therefore, by considering the
state at the end of the process, we obtain that
$$ \sum_{i=1}^n w_i \le \sum_{i=1}^n a_i = \sum_{i=n-m+1}^n a_i \le 2\,\sum_{i=n-m+1}^n w_i.$$
\end{proof}

\begin{theorem}[$d$-to-$d$ games]\label{thm:ddgames}
Let $G$ be a $d$-to-$d$ game for some $d \ge 2$, and assume that $\omega_{\sdp1}(G) = 1$.
Then $\omega^*(G) \geq \frac{1}{20(d-1)}$.
\end{theorem}
\begin{proof}
Fix a solution $\{u_a^s\}$, $\{v_b^t\}$, $z$ to SDP~1 with value $1$ and consider the
strategy of Alice and Bob given by Algorithm~\ref{alg:1}. Our goal is to show that this strategy has success
probability at least $\frac{1}{20(d-1)}$. Clearly, it suffices to show this for any
fixed questions $s,t$. So from now on fix a pair of questions $s,t$, and let $a_{ij} = \ipb{u_i^s, v_j^t}$.
Let $V \subseteq \{1,\ldots,k\}^2$ be the set of allowed answers from the provers, and notice
that $\sum_{i,j \in V} a_{ij} = 1$ and $a_{ij}$ is zero for all $(i,j) \notin V$.
By Corollary~\ref{cor:goodanswers}, the success probability is at least $p_V/2$ where
$$ p_V = \sum_{i,j \in V} a_{ij} \, \frac{a_{ij}}{\max(\sum_{i'} a_{i'j}, \sum_{j'} a_{ij'} )}.$$
Let $V' \subseteq V$ be the set of all pairs $(i,j) \in V$ for which
$\max(\sum_{i'} a_{i'j}, \sum_{j'} a_{ij'}) \le 5(d-1) a_{ij}$. Clearly, $p_V \ge \frac{1}{5(d-1)} \sum_{i,j \in V'} a_{ij}$,
and hence it suffices to lower bound $\sum_{i,j \in V'} a_{ij}$ by $\frac{1}{2}$.

Consider the directed graph on vertex set $V$ defined as follows. We assign the weight $a_{ij}$
to each vertex $(i,j) \in V$ so that the total weight of vertices is $1$.
Let $(i,j)$ be some vertex $V\setminus V'$. If $\sum_{i'} a_{i'j} > 5(d-1) a_{ij}$,
then we put an edge from $(i,j)$
to $(i',j)$ for all $i'$ such that $a_{i'j} > a_{ij}$.
Otherwise, it must be the case
that $\sum_{j'} a_{ij'} > 5(d-1) a_{ij}$, and we proceed similarly,
placing an edge from $(i,j)$ to $(i,j')$ for all $j'$ such that $a_{ij'} > a_{ij}$.

The graph obtained is clearly acyclic. Moreover, the sum of weights of the out-neighbors
of each vertex $(i,j) \in V\setminus V'$ is at least $5(d-1) a_{ij} - (d-1)a_{ij} = 4(d-1)a_{ij}$,
the worst case being when $d-1$ elements in the sum $\sum_{j'} a_{ij'}$
(or $\sum_{i'} a_{i'j}$) are equal to $a_{ij}$. Also, the vertices in $V'$ are exactly
those with outdegree zero, and all indegrees are at most $2(d-1)$. We can therefore
apply Claim~\ref{clm:acycligraph} with $D=2(d-1)$ to obtain that the weight of vertices
in $V'$ is at least $\frac{1}{2}$. This implies that $p_V \ge \frac{1}{10(d-1)}$, as required.
\end{proof}

\section{Parallel Repetition}\label{sec:rep}

In this section we prove our parallel repetition results for unique entangled games.
Given two games $G_1=G(\pi_1,V_1)$ with questions $Q_1$ and answers in $[k_1]$ and $G_2=G(\pi_2,V_2)$ with
questions $Q_2$ and answers in $[k_2]$, we define the product $G_1 \times G_2$ to be a game with questions
$Q_1 \times Q_2$ and answers in $[k_1] \times [k_2]$. The questions are sampled according to the product
distribution $\pi_1 \times \pi_2$. The predicate is the product of  $V_1$ and $V_2$, i.e., the answers are
accepted iff the provers would win each game separately. We denote the $m$-fold product of $G$ with itself by
$G^{  m}$.

The parallel repetition question asks to describe the behavior of the value of $G^m$.
Let us first consider the non-entangled case, which has been investigated
extensively. One easy observation is that $\omega(G^m)$ is lower bounded by $\omega(G)^m$,
as the provers can play each instance of the game independently, using an optimal strategy.
Parallel repetition theorems attempt to provide upper bounds on $\omega(G^m)$.
For general games, the best known result
is Holenstein's tightened version~\cite{Holenstein07} of Raz's parallel repetition theorem~\cite{Raz98}. For
the special case of unique games the best known result is the recent one by Rao
\cite{Rao:parallel}.\footnote{In fact, this result holds for the more general case of \emph{projection games} (see \cite{Rao:parallel} for the definition).}
\begin{theorem}[\cite{Rao:parallel}]\label{Rao:theorem}
Let $G$ be a unique game with value $\omega(G)=1-\eps$. Then for all $m \ge 1$, $(1-\eps)^m \leq \omega(G^{  m})
\leq (1-c \eps^2)^{m}$ where $c>0$ is a universal constant.
\end{theorem}
Somewhat surprisingly, the square in the upper bound is not an artifact of the
proof, but is in fact necessary. There are examples of unique games $G$ (and even XOR games)
for which the upper bound is essentially tight,
namely, $\omega(G)=1-\eps$ but $\omega(G^m)$
is at least $(1-c' \eps^2)^{m}$ for large $m$ and some universal constant $c'>0$~\cite{Raz08,BarakHHRRS08}.

In the entangled case, much less is known.
Clearly, the lower bound $\omega^*(G^m) \ge \omega^*(G)^m$ holds
for the same reason as before.
The only previously known parallel repetition theorem is due to Cleve et al.~\cite{CleveSUU07}, who
prove that the entangled value of XOR games behaves \emph{perfectly} under parallel repetition, meaning
that $\omega^*(G^{ m})=\omega^*(G)^m$. This is in contrast to the behavior of the non-entangled
value of XOR games mentioned above.

\subsection{Our parallel repetition results}

We show the following parallel repetition theorems for unique entangled games.

\begin{theorem}[Parallel repetition for unique games]\label{thm:parallelunique}
Let $G$ be a unique game with entangled value $\omega^*(G)=1-\eps$. Then $(1-\eps)^m \leq \omega^*(G^{  m})
\leq (1-\frac{\eps^2}{16})^m$.
\end{theorem}

In the case of \emph{uniform} unique entangled games we obtain a much stronger
bound (assuming the product game is also uniform).

\begin{theorem}[Parallel repetition for uniform unique games]\label{thm:paralleluniform}
Let $G$ be a  uniform unique game with entangled value $\omega^*(G)=1-\eps$ and such that $G^{  m}$ is also
uniform. Then $(1-\eps)^m \leq \omega^*(G^{  m}) \leq (1-\frac{\eps}{4})^m$.
\end{theorem}

Recall that linear games (cf. Definition~\ref{def:linear}) are uniform.
Moreover, notice that the $m$-th power of a linear game is also linear
(with the group $H^m$). Therefore we obtain the following corollary
of Theorem \ref{thm:paralleluniform}.

 \begin{corollary}
Let $G$ be a  linear game with $\omega^*(G)=1-\eps$. Then $(1-\eps)^m \leq \omega^*(G^{  m}) \leq
(1-\frac{\eps}{4})^m$.
 \end{corollary}

\subsection{Bipartite SDPs and multiplicativity}

To prove our upper bounds on $\omega^*(G^{m})$, we study an SDP
relaxation and show that its value has a certain multiplicative property, as explained below.
The SDPs that we consider here are of a form that we call {\em bipartite} SDP.
These SDPs have two sets of variables, $u_1,\ldots,u_{n_1}$ and $v_1,\ldots,v_{n_2}$;
the goal function only involves inner products between $u$ variables and $v$ variables;
and the constraints are all equality constraints and involve either only $u$
variables or only $v$ variables. More specifically, the bipartite SDP specified
by the $n_1 \times n_2$ matrix $J$, the $n_1 \times n_1$ symmetric matrices $A^1,\ldots,A^{L_1}$,
the $n_2 \times n_2$ symmetric matrices $B^1,\ldots,B^{L_2}$, and the real numbers $a_1,\ldots,a_{L_1}$,
$b_1,\ldots,b_{L_2}$ is given by
\begin{align}
\textbf{Maximize: \qquad} & \textstyle \sum_{i=1,j=1}^{n_1,n_2} J_{ij}\, \ip{u_i, v_j} \nonumber \\
\textbf{Subject to: \qquad} & \textstyle \sum_{i,j=1}^{n_1} A^{l}_{ij} \ip{u_i, u_j}=a_{l} \text{~for~} l=1,\ldots,L_1 \label{eq:bipartitesdp} \\
   & \textstyle \sum_{i,j=1}^{n_2} B^{l}_{ij} \ip{v_i, v_j} = b_{l} \text{~for~} l=1,\ldots,L_2. \nonumber
\end{align}

We now define the \emph{bipartite product} $S \bprod S'$ of two bipartite SDPs $S$ and $S'$.
Assume $S$ has $n_1+n_2$ variables and $L_1+L_2$ constraints, and
is specified by $J,A^{l},B^{l},a_l$ and $b_l$, and similarly for $S'$.
Then $S \bprod S'$ is the bipartite SDP over $n_1 n'_1 + n_2 n'_2$ variables
and $L_1 L'_1 + L_2 L'_2$ constraints given by $J \otimes J'$, the matrices $A^{l} \otimes A'^{l'}$
and $B^{l} \otimes B'^{l'}$, and the numbers $a_l a'_{l'}$ and $b_l b'_{l'}$.
More specifically, the product SDP $S \bprod S'$ is given by
\begin{align}
\textbf{Maximize:\qquad}
   & \textstyle \sum_{i=1,j=1,i'=1,j'=1}^{n_1,n_2,n'_1,n'_2} J_{ij} J'_{i'j'}\, \ip{u_{ii'}, v_{jj'}} \nonumber \\
\textbf{Subject to:\qquad}
   & \textstyle \sum_{i,j=1,i',j'=1}^{n_1,n'_1} A^{l}_{ij} A'^{l'}_{i'j'} \ip{u_{ii'}, u_{jj'}}=a_l a'_{l'} \text{~for~} l=1,\ldots,L_1,~l'=1,\ldots,L'_1 \label{eq:bipartiteproductsdp} \\
   & \textstyle \sum_{i,j=1,i',j'=1}^{n_2,n'_2} B^{l}_{ij} B'^{l'}_{i'j'}  \ip{v_{ii'}, v_{jj'}} = b_{l} b'_{l'} \text{~for~} l=1,\ldots,L_2,~l'=1,\ldots,L'_2. \nonumber
\end{align}

{}From the construction of the bipartite product it is obvious that, given any feasible solution $\{u_i\},\{v_j\}$
of $S$ and any feasible solution $\{u'_{i'}\},\{v'_{j'}\}$ of $S'$,
we can construct a feasible solution $\{u_i\otimes u'_{i'}\},\{v_j\otimes
v'_{j'}\}$ of $S \bprod S'$ whose value is the product of the two values.
By taking optimal solutions, this immediately implies that $\omega_{S} \omega_{S'} \leq
\omega_{S \bprod S'}$, where $\omega$ indicates the value of the SDP.
However, it is not a priori clear that equality holds here.
It could well be the case that $S \bprod S'$ has solutions that do not have a tensor product
structure and give a higher value. Luckily, we now show that equality \emph{does}
hold under a very mild condition on the SDPs, namely, that they are strictly feasible, which
in our case simply means that there is a feasible solution in which all vectors
have positive length.

\begin{theorem}\label{thm:MS}
For any two strictly feasible bipartite SDPs $S$ and $S'$, $\omega_{S \bprod S'} = \omega_{S} \omega_{S'}$.
\end{theorem}

This theorem can be shown to follow from the work of Mittal and Szegedy
\cite{MittalSzegedy:SDPtensor}, as we describe below. However, for completeness, we first give a
self-contained proof, which is a modification of the arguments in~\cite{MittalSzegedy:SDPtensor}.
The remainder of this section is dedicated to the proof of Theorem \ref{thm:MS}. Later sections are
independent of this proof, so the reader can safely jump to Section \ref{sec:proofrep} on first
reading.

\begin{proof}
The main tool in the proof is SDP duality (see, e.g., \cite{vandenberghe&boyd:sdp}). When applied
to a strictly feasible bipartite SDP as in Eq.~\eqref{eq:bipartitesdp}, it says that the value of
the SDP is equal to that of the following program over real variables $x_1,\ldots,x_{L_1},
y_1,\ldots,y_{L_2}$.
\begin{align}
\textbf{Minimize:\qquad}  &
   \sum_{l=1}^{L_1} x_{l}\, a_l + \sum_{l=1}^{L_2} y_{l}\, b_l \nonumber \\
\textbf{Subject to:\qquad}&
   \left(%
\begin{array}{cc}
  \sum_{l=1}^{L_1} x_l A^{l}    &  0 \\
  0   &  \sum_{l=1}^{L_2} y_l B^{l} \\
\end{array}%
\right) \succeq \left(
\begin{array}{cc}
  0    &  J/2 \\
  J^T/2   &  0 \\
\end{array}
\right) \label{eq:matrixindualsdp}
\end{align}
Moreover, even if the bipartite SDP is not strictly feasible, the value of the program above is
always an upper bound on the value of the SDP.

Using our assumption that $S$ is strictly feasible, we obtain that there are real numbers
$x_1,\ldots,x_{L_1}$, $y_1,\ldots,y_{L_2}$ satisfying that $\sum_{l=1}^{L_1} x_{l}\, a_l +
\sum_{l=1}^{L_2} y_{l}\, b_l = \omega_S$ and that Eq.~\eqref{eq:matrixindualsdp} holds. An
important observation is that we can assume without loss of generality that the two sums in the
goal function are equal, i.e., $\sum_{l=1}^{L_1} x_{l}\, a_l = \sum_{l=1}^{L_2} y_{l}\, b_l =
\omega_S/2$. To see this, notice that for any fixed $\alpha > 0$ and any feasible solution, if we
multiply all the $x$ variables by $\alpha$ and all the $y$ variables by $1/\alpha$, then
Eq.~\eqref{eq:matrixindualsdp} remains valid. This follows from the fact that for any symmetric
matrices $X,Y$ and any matrix $R$, if $X \succeq Y$ then also $R^T X R \succeq R^T Y R$; in our
case $R$ is the diagonal matrix with
$\sqrt{\alpha},\ldots,\sqrt{\alpha},1/\sqrt{\alpha},\ldots,1/\sqrt{\alpha}$ on the diagonal. We
will also need later the observation that Eq.~\eqref{eq:matrixindualsdp} remains valid if we
multiply its right hand side by $-1$; this follows by taking $R$ to be the diagonal matrix with
$1,\ldots,1,-1,\ldots,-1$ on the diagonal. Following the exact same reasoning for $S'$ leads to
real numbers $x'_1,\ldots,x'_{L'_1}$, $y'_1,\ldots,y'_{L'_2}$ satisfying that $\sum_{l=1}^{L'_1}
x'_{l}\, a'_l = \sum_{l=1}^{L'_2} y'_{l}\, b'_l = \omega_{S'}/2$ and that the inequality analogous
to Eq.~\eqref{eq:matrixindualsdp} holds (both as is, and when multiplying its right hand side by
$-1$).

Recall that our goal is to prove $\omega_{S \bprod S'} \le \omega_{S} \omega_{S'}$ (since the lower
bound is easy). By applying the discussion above to the bipartite SDP $S\bprod S'$ as given in
Eq.~\eqref{eq:bipartiteproductsdp}, we see that $\omega_{S \bprod S'}$ is at most (in fact, equal
to) the value of the following program over $L_1 L'_1 + L_2 L'_2$ real variables.
\begin{align}
\textbf{Minimize:\qquad}  &
   \sum_{l=1,l'=1}^{L_1,L'_1} x_{ll'}\, a_l a'_{l'} + \sum_{l=1,l'=1}^{L_2,L'_2} y_{ll'}\, b_l b'_{l'} \label{eq:goalinproductdualsdp} \\
\textbf{Subject to:\qquad}&
   \left(%
\begin{array}{cc}
  \sum_{l=1,l'=1}^{L_1,L'_1} x_{ll'} A^{l} \otimes A'^{l'}    &  0 \\
  0   &  \sum_{l=1,l'=1}^{L_2,L'_2} y_{ll'} B^{l} \otimes B'^{l'} \\
\end{array}%
\right) \succeq
   \left(%
\begin{array}{cc}
  0  &  J \otimes J' /2 \\
  J^T \otimes J'^T /2   &  0 \\
\end{array}%
\right) . \label{eq:matrixinproductdualsdp}
\end{align}

Consider the assignment given by $x_{ll'}=2 x_l x'_{l'}$, $y_{ll'}=2 y_l y'_{l'}$. Notice that with
this assignment, the goal function in Eq.~\eqref{eq:goalinproductdualsdp} is exactly $\omega_{S}
\omega_{S'}$ hence in the following it suffices to show that for this assignment,
Eq.~\eqref{eq:matrixinproductdualsdp} holds. To prove this, we will use the following simple but
crucial claim.

\begin{claim}
If $X,Y,X',Y'$ are symmetric matrices for which $X \succeq Y$, $X \succeq -Y$, $X' \succeq Y'$, and
$X' \succeq -Y'$, then also $X\otimes X' \succeq Y \otimes Y'$.
\end{claim}
\begin{proof}
Our assumption says that $X - Y \succeq 0$, $X + Y \succeq 0$, $X' - Y'\succeq 0$, and $X' +
Y'\succeq 0$. Since both the tensor product and the sum of positive semidefinite matrices are
positive semidefinite, we obtain that the following matrix is positive semidefinite
$$ (X-Y) \otimes (X'+Y') + (X+Y) \otimes (X'-Y') = 2 (X \otimes X' - Y \otimes Y'),$$
and the claim follows.
\end{proof}

Applying this claim to the inequality in Eq.~\eqref{eq:matrixindualsdp} and to the analogous
inequality for $S'$ (recalling our observation above that both inequalities remain valid when the
right hand side is multiplied by $-1$), we obtain that

\begin{align*}
& \left(%
\begin{array}{cccc}
  \big(\sum_{l=1}^{L_1} x_l A^{l}\big)\otimes\big(\sum_{l'=1}^{L'_1} x'_{l'} A'^{l'}\big)     &  0 & 0 & 0\\
   0 & \qquad * \qquad    &  0 & 0 \\
  0 & 0   &  \qquad *\qquad  & 0  \\
  0 & 0   &  0 & \big(\sum_{l=1}^{L_2} y_l B^{l} \big) \otimes\big(\sum_{l'=1}^{L'_2} y'_{l'} B'^{l'}\big)  \\
\end{array}%
\right)  \\
& \succeq \left(
\begin{array}{cccc}
  0    &  0 & 0 & J \otimes J' /4 \\
  0    &  0 & J \otimes J'^T /4 & 0 \\
  0    &  J^T \otimes J' /4 & 0 & 0 \\
  J^T \otimes J'^T /4   &  0 & 0 & 0\\
\end{array}
\right)
\end{align*}
where the stars indicate expressions omitted due to lack of space. We obtain the desired
inequality, Eq.~\eqref{eq:matrixinproductdualsdp}, by multiplying both sides of the inequality by
$2$ and taking the rows and columns corresponding to the first and fourth blocks, which is an
operation that preserves semidefinite inequalities. This completes the proof of Theorem
\ref{thm:MS}.
\end{proof}

We now outline how Theorem \ref{thm:MS} can be derived from the work of Mittal and
Szegedy~\cite{MittalSzegedy:SDPtensor} on multiplicative properties of SDPs. In this work they
describe several classes of SDPs (which include our bipartite SDPs as a special case) for which
$\omega_{S \otimes S'} = \omega_{S} \omega_{S'}$ where $S \otimes S'$ denotes the \emph{tensor
product} of the two SDPs. The tensor product of SDPs is defined similarly to our bipartite product,
except that it ignores the bipartite structure of the SDP and can be applied to any SDP with
equality constraints; see~\cite{MittalSzegedy:SDPtensor} for the details. As a result, when
applying the tensor product to the two bipartite SDPs as in Eq.~\eqref{eq:bipartitesdp}, we end up with
$n_1 n'_2 + n_2 n'_1$ additional `cross variables' (so the total number of variables is $(n_1 +
n_2)(n'_1+n'_2)$) as well as $L_1 L'_2 + L_2 L'_1$ extra `cross constraints' on the cross variables
(so the total number of constraints is $(L_1 +L_2)(L'_1 + L'_2)$). Denoting the cross variables by
$\{w_{ij}\}$ and $\{z_{ij}\}$, the tensor product of two bipartite SDPs is given by
\begin{align*}
\textbf{Maximize:\qquad} &
    \textstyle \frac{1}{2} \left( \sum_{i=1,j=1,i'=1,j'=1}^{n_1,n_2,n'_1,n'_2} J_{ij} J'_{i'j'}\, \ip{u_{ii'}, v_{jj'}} +
     \sum_{i=1,j=1,i'=1,j'=1}^{n_1,n_2,n'_2,n'_1} J_{ij} J'_{j'i'}\, \ip{w_{ii'}, z_{jj'}} \right)
    \\
\textbf{Subject to:\qquad}
   & \textstyle \sum_{i,j=1,i',j'=1}^{n_1,n'_1} A^{l}_{ij} A'^{l'}_{i'j'} \ip{u_{ii'}, u_{jj'}}=a_l a_{l'} \text{~for~} l=1,\ldots,L_1,~ l'=1,\ldots,L'_1 \\
   & \textstyle \sum_{i,j=1,i',j'=1}^{n_2,n'_2} B^{l}_{ij} B'^{l'}_{i'j'} \ip{v_{ii'}, v_{jj'}}=b_l b_{l'} \text{~for~} l=1,\ldots,L_2,~ l'=1,\ldots,L'_2 \\
   & \textstyle \sum_{i,j=1,i',j'=1}^{n_1,n'_2} A^{l}_{ij} B'^{l'}_{i'j'} \ip{w_{ii'}, w_{jj'}}=a_l b_{l'} \text{~for~} l=1,\ldots,L_1,~ l'=1,\ldots,L'_2 \\
   & \textstyle \sum_{i,j=1,i',j'=1}^{n_2,n'_1} B^{l}_{ij} A'^{l'}_{i'j'} \ip{z_{ii'}, z_{jj'}}=b_l a_{l'} \text{~for~} l=1,\ldots,L_2,~ l'=1,\ldots,L'_1.
\end{align*}
The reason for the factor half in the goal function is that in standard SDP form the goal
function must be symmetric, hence the goal function of a bipartite SDP is better thought of as
$\frac{1}{2}(\sum_{i=1,j=1}^{n_1,n_2} J_{ij}\, \ip{u_i, v_j} + \sum_{i=1,j=1}^{n_1,n_2} J_{ij}\, \ip{v_j, u_i})$.
This also explains why the goal function in the tensor SDP involves cross terms.

Notice that there are no constraints involving both the cross variables and the normal variables,
and also that the goal function does not contain any inner products of cross variables with
normal variables.
This implies that the value of this SDP, $\omega_{S \otimes S'}$, is equal to
$\frac{1}{2}(\omega_{S \bprod S'}+ \omega_{S \bprod' S'})$
where $S \bprod' S'$ represents the SDP on the cross variables.
By the theorem of Mittal and Szegedy~\cite{MittalSzegedy:SDPtensor} mentioned above,
$\omega_{S \otimes S'} = \omega_{S} \omega_{S'}$ (this uses the assumption
that both $S$ and $S'$ are strictly feasible). Moreover, by taking
the tensor solution, it is easy to see that $\omega_{S \bprod' S'} \ge \omega_{S} \omega_{S'}$.
This implies that $\omega_{S \bprod S'} \le \omega_{S} \omega_{S'}$, as required.

\subsection{Proof of parallel repetition}\label{sec:proofrep}

Note that the two SDPs in Section \ref{sec:relaxation} are not bipartite. Therefore, in order to apply
Theorem \ref{thm:MS}, we define two bipartite SDPs, SDP~\ref{sdp:3} and SDP~\ref{sdp:4}.

\setcounter{program}{2}

\begin{program}
\caption{}\label{sdp:3}
\begin{tabular}{p{0.24 \textwidth}p{0.7 \textwidth}}
\textbf{Maximize:} & $\sum_{abst} \pi(s,t) V(a,b\,|\,s,t) \ip{u_a^s, v_b^t}$\\
\textbf{Subject to:} & $\forall s,~\forall a \neq b,~ \ip{u_a^s,u_b^s}=0$ and $\forall t,~\forall a \neq b,~ \ip{v_a^t,v_b^t}=0$\\
& $\forall s,~\sum_{a}\ip{u_a^s,u_a^s} =1$ and $\forall t,~\sum_{b}\ip{v_b^t,v_b^t} =1$\\
\end{tabular}
\end{program}

It is easy to see that SDP~\ref{sdp:3} is a relaxation of SDP~\ref{sdp:1}, and hence
for any game $G$ its value satisfies $\omega_{\sdp\text{\ref{sdp:3}}}(G) \ge \omega_{\sdp\text{\ref{sdp:1}}}(G) \ge \omega^*(G) $.
For \emph{uniform} games, we will work with SDP~\ref{sdp:4}.
It is easy to see that it is a relaxation of SDP~2 and hence for any uniform game $G$ its value
satisfies $\omega_{\sdp\text{\ref{sdp:4}}}(G) \ge \omega_{\sdp2}(G) \ge \omega^*(G)$.

\begin{program}
\caption{}\label{sdp:4}
\begin{tabular}{p{0.24 \textwidth}p{0.7 \textwidth}}
\textbf{Maximize:} & $\sum_{abst} \pi(s,t) V(a,b\,|\,s,t) \ip{u_a^s, v_b^t}$\\
\textbf{Subject to:} & $\forall s, a, b,~ \ip{u_a^s,u_b^s}=\frac{1}{k}\delta_{a,b}$
and $\forall t, a, b,~ \ip{v_a^t,v_b^t}=\frac{1}{k}\delta_{a,b}$\\
\end{tabular}
\end{program}

We now prove our parallel repetition theorem for the case of uniform unique games
(whose power is also uniform).

\begin{proof}[ of Theorem \ref{thm:paralleluniform}]
The lower bound is obvious. For the upper bound, recall that in Remark~\ref{remark:2}
we observed that the only constraints of SDP~2 used in the proof of Theorem~\ref{thm:lingames}
are exactly those that now appear in SDP~\ref{sdp:4}. Hence, the proof of Theorem \ref{thm:lingames}
holds word by word if we replace SDP~2 by SDP~\ref{sdp:4}, and we obtain that
our assumption $\omega^*(G) = 1-\eps$ implies that
$\omega_{\sdp\text{\ref{sdp:4}}}(G) \le 1-\frac{\eps}{4}$.

The next crucial observation is that for any two games $G$ and $G'$ we
have that SDP\ref{sdp:4}$(G)\bprod $SDP\ref{sdp:4}$(G')$ is identical to SDP\ref{sdp:4}$(G \times G')$.
To see this, notice that both are equal to
\begin{align*}
\textbf{Maximize:\qquad} & \textstyle \sum_{aa'bb'ss'tt'} \pi(s,t)\pi'(s',t') V(a,b\,|\,s,t) V'(a',b'\,|\,s',t') \ipb{u_{aa'}^{ss'}, v_{bb'}^{tt'}} \\
\textbf{Subject to:\qquad} & \forall s, s', a, a', b, b'~ \ipb{u_{aa'}^{ss'},u_{bb'}^{ss'}}=\textstyle \frac{1}{k^2}\delta_{aa',bb'}
\text{~and~} \\
& \forall t, t', a, a', b,b',~ \ipb{v_{aa'}^{tt'},v_{bb'}^{tt'}}=\textstyle \frac{1}{k^2}\delta_{aa',bb'}.
\end{align*}
In particular, this implies that SDP\ref{sdp:4}$(G^m)$ is identical to SDP\ref{sdp:4}$(G)^{\bprod m}$.

Next, observe that for any game $G$, SDP\ref{sdp:4}(G)
is strictly feasible, as can be seen by taking, say, all vectors to be orthogonal of norm $1/\sqrt{k}$.
This allows us to apply Theorem~\ref{thm:MS} and obtain that
$\omega_{\sdp\text{\ref{sdp:4}}}(G^{ m})$, which is the same as the value of SDP\ref{sdp:4}$(G)^{\bprod m}$,
is equal to $\omega_{\sdp\text{\ref{sdp:4}}}(G)^m$.
Putting all this together, we get
$$\omega^*(G^{m}) \leq \omega_{\sdp\text{\ref{sdp:4}}}(G^{ m})=\omega_{\sdp\text{\ref{sdp:4}}}(G)^m \leq \bigg(1-\frac{\eps}{4}\bigg)^m,$$
where the first inequality uses the assumption that $G^m$ is uniform.
\end{proof}

We now turn to the case of general unique games. The situation here is more involved, for two
reasons. The first is that Theorem~\ref{thm:uniqgames} uses all the constraints in SDP~\ref{sdp:1}
and we have to modify it in order to derive a similar (but significantly weaker) statement for
SDP~\ref{sdp:3}. This is done in Lemma \ref{lemma:uniquemodified} below. The second difficulty is
that the product of SDP~\ref{sdp:3} is not exactly the same as SDP~\ref{sdp:3} of the product game.
We resolve this in the proof of Theorem \ref{thm:parallelunique} below by noting that the former is a
relaxation of the latter.

\begin{lemma}\label{lemma:uniquemodified}
Let $G$ be a unique game. Suppose that $\omega_{\sdp\text{\ref{sdp:3}}}(G) = 1 -\eps$. Then $\omega^*(G) \geq 1  -
2\sqrt{2\eps} - 4\eps$.
\end{lemma}
\begin{proof}
Fix a solution $\{u_a^s\}$, $\{v_b^t\}$ to SDP~\ref{sdp:3} with value $1-\eps$ and consider the strategy of
Alice and Bob given by Algorithm~\ref{alg:1}. Our goal now is to show that this strategy has success
probability at least $1 - 2\sqrt{2\eps}-4\eps$. Using the fact that the square root of the expectation is an
upper bound on the expectation of the square root, it is easy to see that it suffices to show that for any
questions $s,t$, the success probability of Alice and Bob on these questions is at least
$$1 - 2\sqrt{2}\sqrt{1-\sum_{ab} V(a,b\,|\,s,t) \ip{u_a^s, v_b^t}} - 4(1-\sum_{ab} V(a,b\,|\,s,t) \ip{u_a^s, v_b^t}) .$$
With the notation of the proof of Theorem \ref{thm:uniqgames}, assuming $\sum_i \ip{u_i,v_i}\geq 1-\tilde \eps$ we
have to show  $p'_{\rm succ} \ge 1-\sqrt{2 \tilde \eps}-2\tilde \eps$. \footnote{Note that $\tilde \eps$ must
be non-negative because $1-\tilde \eps \leq F \leq 1$, so $\sqrt{\tilde \eps}$ is well defined.} As in the proof of
Theorem \ref{thm:uniqgames} we can define $p''_{succ}$ and obtain $p''_{succ} \geq 1-2\tilde \eps$. We used
the non-bipartite constraints of SDP~\ref{sdp:1} in the proof of Theorem \ref{thm:uniqgames} only when
bounding $p''_{\rm succ} - p'_{\rm succ}$. We now give a weaker bound on this quantity, which only uses the
constraints of SDP~\ref{sdp:3}.
\begin{align*}
p''_{\rm succ} - p'_{\rm succ}&\le   \frac{1}{2} \sum_i | \|u_i\|^2 -  \|v_i\|^2 | \\
  &=  \frac{1}{2} \sum_i | \|u_i\| -  \|v_i\| |\cdot | \|u_i\| +  \|v_i\| | \\
  &\leq  \frac{1}{2} \sqrt{\sum_i ( \|u_i\| -  \|v_i\|)^2 }\sqrt{\sum_i ( \|u_i\| +  \|v_i\|)^2 } \\
  & = \frac{1}{2} \sqrt{(2-2\sum_i \|u_i\| \|v_i\|)(2+2\sum_i \|u_i\| \|v_i\|)}\\
  &= \sqrt{1-\big(\sum_i \|u_i\| \|v_i\|\big)^2} \leq \sqrt{1-\big(\sum_i \ip{u_i,v_i}\big)^2} \leq \sqrt{2 \tilde \eps}.
\end{align*}
\end{proof}

Now we can turn to the proof of Theorem \ref{thm:parallelunique}.

\begin{proof}[ of Theorem \ref{thm:parallelunique}]
The lower bound is again obvious. For the upper bound, Lemma
\ref{lemma:uniquemodified} implies that if $\omega^*(G)=1-\eps$ then
$\omega_{\sdp\text{\ref{sdp:3}}}(G)\leq 1-\frac{\eps^2}{16}$.

Next, we consider for any two games $G$ and $G'$,
SDP\ref{sdp:3}$(G)\bprod $SDP\ref{sdp:3}$(G')$ and SDP\ref{sdp:3}$(G \times G')$.
These two SDPs are no longer identical. Instead, as we shall soon see, the former is a relaxation of
the latter. This implies that
$\omega_{\sdp\text{\ref{sdp:3}}}(G^{ m})$ is less than or equal to the value of SDP\ref{sdp:3}$(G)^{\bprod m}$.

The proof now continues as before. We observe that for any game $G$, SDP\ref{sdp:3}(G)
is strictly feasible, for the same reason as before.
This allows us to apply Theorem~\ref{thm:MS} and obtain that
the value of SDP\ref{sdp:3}$(G)^{\bprod m}$
is equal to $\omega_{\sdp\text{\ref{sdp:3}}}(G)^m$.
Putting all this together, we get
$$\omega^*(G^{m})
    \le \omega_{\sdp\text{\ref{sdp:3}}}(G^{m})
    \le \omega_{\sdp\text{\ref{sdp:3}}}(G)^m
    \le \bigg(1-\frac{\eps^2}{16}\bigg)^m.$$

It remains to show why for any two games $G$ and $G'$,
SDP\ref{sdp:3}$(G)\bprod $SDP\ref{sdp:3}$(G')$ is a relaxation of
SDP\ref{sdp:3}$(G \times G')$. The set of variables and the goal function
are identical in both SDPs. The only difference is in the constraints.
In the former SDP, the constraints are
\begin{align*}
& \forall s, s', \forall a \neq b, \forall a' \neq b',~ \ipb{u_{aa'}^{ss'},u_{bb'}^{ss'}}=0 \\
& \forall s, s', \forall a \neq b,~ \sum_{a'}\ipb{u_{aa'}^{ss'},u_{ba'}^{ss'}} = 0 \\
& \forall s, s', \forall a' \neq b',~ \sum_{a}\ipb{u_{aa'}^{ss'},u_{ab'}^{ss'}} = 0 \\
& \forall s, s',~ \sum_{a,a'}\ipb{u_{aa'}^{ss'},u_{aa'}^{ss'}} =1
\end{align*}
and similarly for the $v$ variables. In the latter SDP, the constraints are
\begin{align*}
& \forall s, s', \forall (a,a') \neq (b,b'),~ \ipb{u_{aa'}^{ss'},u_{bb'}^{ss'}}=0 \\
& \forall s, s', \sum_{a,a'}\ipb{u_{aa'}^{ss'},~ u_{aa'}^{ss'}} =1
\end{align*}
and similarly for the $v$ variables. To complete the proof, notice that
the last constraint is the same in both SDPs, and that the first constraint
in the latter SDP implies the first three constraints in the former SDP.
\end{proof}

\begin{remark}
As we saw above, for any two games $G$ and $G'$,
the value of SDP\ref{sdp:3}$(G)\bprod $SDP\ref{sdp:3}$(G')$ is greater
than or equal to that of SDP\ref{sdp:3}$(G \times G')$. In fact,
the two are equal. To see this, notice that by Theorem~\ref{thm:MS} the optimum of
the former is attained by taking the tensor of the optimum
solutions to SDP\ref{sdp:3}$(G)$ and SDP\ref{sdp:3}$(G')$,
and that this tensor solution is also a feasible solution
to SDP\ref{sdp:3}$(G \times G')$.
\end{remark}

\subsection*{Acknowledgments}

We thank Stephanie Wehner for useful discussions.

\newcommand{\etalchar}[1]{$^{#1}$}


\begin{thebibliography}{PGWP{\etalchar{+}}08}
\expandafter\ifx\csname urlstyle\endcsname\relax
  \providecommand{\doi}[1]{doi:\discretionary{}{}{}#1}\else
  \providecommand{\doi}{doi:\discretionary{}{}{}\begingroup
  \urlstyle{rm}\Url}\fi

\bibitem[AKK{\etalchar{+}}08]{unique:expander}
S.~Arora, S.~Khot, A.~Kolla, D.~Steurer, M.~Tulsiani, and N.~Vishnoi.
\newblock Unique games on expanding constraint graphs are easy.
\newblock In \emph{Proc. 40th ACM Symp. on Theory of Computing}, pages 21--28.
  2008.

\bibitem[ALM{\etalchar{+}}98]{ALMSS98}
S.~Arora, C.~Lund, R.~Motwani, M.~Sudan, and M.~Szegedy.
\newblock Proof verification and the hardness of approximation problems.
\newblock \emph{J. ACM}, 45(3):501--555, 1998.

\bibitem[AS98]{AS98}
S.~Arora and S.~Safra.
\newblock Probabilistic checking of proofs: a new characterization of {NP}.
\newblock \emph{J. ACM}, 45(1):70--122, 1998.

\bibitem[Bel64]{Bell:64a}
J.~S. Bell.
\newblock On the \uppercase{E}instein-\uppercase{P}odolsky-\uppercase{R}osen
  paradox.
\newblock \emph{Physics}, 1:195--200, 1964.

\bibitem[BHH{\etalchar{+}}08]{BarakHHRRS08}
B.~Barak, M.~Hardt, I.~Haviv, A.~Rao, O.~Regev, and D.~Steurer.
\newblock Rounding parallel repetitions of unique games.
\newblock In \emph{Proc. 49th IEEE Symp. on Foundations of Computer Science},
  pages 374--383. 2008.

\bibitem[BM04]{Buhrman04}
H.~Buhrman and S.~Massar.
\newblock Causality and {C}irel'son bounds, 2004.
\newblock Quant-ph/0409066.

\bibitem[BV04]{BoydV04}
S.~Boyd and L.~Vandenberghe.
\newblock \emph{Convex optimization}.
\newblock Cambridge University Press, Cambridge, 2004.

\bibitem[CHSH69]{Clauser:69a}
J.~F. Clauser, M.~A. Horne, A.~Shimony, and R.~A. Holt.
\newblock Proposed experiment to test local hidden-variable theories.
\newblock \emph{Phys. Rev. Lett.}, 23:880--884, 1969.

\bibitem[CHTW04]{CleveHTW04}
R.~Cleve, P.~H{\o}yer, B.~Toner, and J.~Watrous.
\newblock Consequences and limits of nonlocal strategies.
\newblock In \emph{Proc. 19th IEEE Conference on Computational Complexity},
  pages 236--249. 2004.

\bibitem[Cir80]{Tsirelson:80a}
B.~S. Cirel'son.
\newblock Quantum generalizations of \uppercase{B}ell's inequality.
\newblock \emph{Lett. Math. Phys.}, 4:93--100, 1980.

\bibitem[CKK{\etalchar{+}}06]{ChawlaKKRS06}
S.~Chawla, R.~Krauthgamer, R.~Kumar, Y.~Rabani, and D.~Sivakumar.
\newblock On the hardness of approximating multicut and sparsest-cut.
\newblock \emph{Comput. Complexity}, 15(2):94--114, 2006.

\bibitem[CMM06a]{CMM}
M.~Charikar, K.~Makarychev, and Y.~Makarychev.
\newblock Near-optimal algorithms for unique games.
\newblock In \emph{Proc. 38th ACM Symp. on Theory of Computing}, pages
  205--214. 2006.

\bibitem[CMM06b]{ChlamtacMM06}
E.~Chlamtac, K.~Makarychev, and Y.~Makarychev.
\newblock How to play unique games using embeddings.
\newblock In \emph{Proc. 47th IEEE Symp. on Foundations of Computer Science},
  pages 687 -- 696. 2006.

\bibitem[CSUU07]{CleveSUU07}
R.~Cleve, W.~Slofstra, F.~Unger, and S.~Upadhyay.
\newblock Perfect parallel repetition theorem for quantum \uppercase{XOR} proof
  systems.
\newblock In \emph{Proc. 22nd IEEE Conference on Computational Complexity},
  pages 109--114. 2007.

\bibitem[DMR06]{DinurMR06}
I.~Dinur, E.~Mossel, and O.~Regev.
\newblock Conditional hardness for approximate coloring.
\newblock In \emph{Proc. 38th ACM Symp. on Theory of Computing}, pages
  344--353. 2006.

\bibitem[DS05]{DinurS05}
I.~Dinur and S.~Safra.
\newblock On the hardness of approximating minimum vertex cover.
\newblock \emph{Ann. of Math. (2)}, 162(1):439--485, 2005.

\bibitem[EPR35]{EinsteinPR35}
A.~Einstein, P.~Podolsky, and N.~Rosen.
\newblock Can quantum-mechanical description of physical reality be considered
  complete?
\newblock \emph{Phys. Rev.}, 47:777--780, 1935.

\bibitem[Fei98]{Feige98}
U.~Feige.
\newblock A threshold of {$\ln n$} for approximating set cover.
\newblock \emph{J. ACM}, 45(4):634--652, 1998.

\bibitem[FL92]{FL}
U.~Feige and L.~Lov\'asz.
\newblock Two-prover one-round proof systems: Their power and their problems.
\newblock In \emph{Proc. 24th ACM Symp. on Theory of Computing}, pages
  733--741. 1992.

\bibitem[GT06]{GT06}
A.~Gupta and K.~Talwar.
\newblock Approximating unique games.
\newblock In \emph{Proc. 17th Annual ACM-SIAM Symposium on Discrete
  Algorithms}, pages 99--106. 2006.

\bibitem[GW95]{GoemansW95}
M.~X. Goemans and D.~P. Williamson.
\newblock Improved approximation algorithms for maximum cut and satisfiability
  problems using semidefinite programming.
\newblock \emph{J. Assoc. Comput. Mach.}, 42(6):1115--1145, 1995.

\bibitem[H{\aa}s99]{Hastad-clique}
J.~H{\aa}stad.
\newblock Clique is hard to approximate within {$n^{1-\epsilon}$}.
\newblock \emph{Acta Math.}, 182(1):105--142, 1999.

\bibitem[H{\aa}s01]{Hastad01}
J.~H{\aa}stad.
\newblock Some optimal inapproximability results.
\newblock \emph{J. ACM}, 48(4):798--859, 2001.

\bibitem[Hol07]{Holenstein07}
T.~Holenstein.
\newblock Parallel repetition: simplifications and the no-signaling case.
\newblock In \emph{Proc. 39th ACM Symp. on Theory of Computing}, pages
  411--419. 2007.

\bibitem[IKM09]{ItoKM09}
T.~Ito, H.~Kobayashi, and K.~Matsumoto.
\newblock Oracularization and two-prover one-round interactive proofs against
  nonlocal strategies.
\newblock In \emph{Proc. 24th IEEE Conference on Computational Complexity},
  pages 217--228. 2009.

\bibitem[Kho02]{Khot02}
S.~Khot.
\newblock On the power of unique 2-prover 1-round games.
\newblock In \emph{Proc. 34th ACM Symp. on Theory of Computing}, pages
  767--775. 2002.

\bibitem[KKM{\etalchar{+}}08]{kempe:immunization}
J.~Kempe, H.~Kobayashi, K.~Matsumoto, B.~Toner, and T.~Vidick.
\newblock Entangled games are hard to approximate.
\newblock In \emph{Proc. 49th IEEE Symp. on Foundations of Computer Science},
  pages 447--456. 2008.

\bibitem[KKMO07]{KhotKMO07}
S.~Khot, G.~Kindler, E.~Mossel, and R.~O'Donnell.
\newblock Optimal inapproximability results for {MAX-CUT} and other 2-variable
  {CSP}s?
\newblock \emph{SIAM J. Comput.}, 37(1):319--357, 2007.

\bibitem[KR08]{KhotR03}
S.~Khot and O.~Regev.
\newblock Vertex cover might be hard to approximate to within $2-\varepsilon$.
\newblock \emph{Journal of Computer and System Sciences (JCSS)},
  74(3):335--349, 2008.

\bibitem[KV05]{KhotV05}
S.~Khot and N.~K. Vishnoi.
\newblock The unique games conjecture, integrality gap for cut problems and
  embeddability of negative type metrics into $l_{\mbox{1}}$.
\newblock In \emph{Proc. 46th IEEE Symp. on Foundations of Computer Science},
  pages 53--62. 2005.

\bibitem[Mas05]{Masanes05}
L.~Masanes.
\newblock Extremal quantum correlations for {$N$} parties with two dichotomic
  observables per site, 2005.
\newblock Quant-ph/0512100.

\bibitem[MM09]{MakarychevM09}
K.~Makarychev and Y.~Makarychev.
\newblock How to play unique games on expanders, 2009.
\newblock ArXiv:0903.0367.

\bibitem[MS07]{MittalSzegedy:SDPtensor}
R.~Mittal and M.~Szegedy.
\newblock Product rules in semidefinite programming.
\newblock In \emph{Proc. 16th Fund. Computation Theory (FCT)}, pages 435--445.
  2007.

\bibitem[NC00]{NC}
M.~A. Nielsen and I.~L. Chuang.
\newblock \emph{Quantum Computation and Quantum Information}.
\newblock Cambridge University Press, New York, 2000.

\bibitem[NPA07]{NavascuesPA07}
M.~Navascues, S.~Pironio, and A.~Ac{\'\i}n.
\newblock Bounding the set of quantum correlations.
\newblock \emph{Phys. Rev. Lett.}, 98(1):010401, 2007.

\bibitem[PGWP{\etalchar{+}}08]{perez-garcia:_unboun}
D.~Perez-Garcia, M.~Wolf, C.~Palazuelos, I.~Villanueva, and M.~Junge.
\newblock Unbounded violation of tripartite \uppercase{B}ell inequalities.
\newblock \emph{Comm. Math. Phys.}, 279(2):455--486, 2008.

\bibitem[Rao08]{Rao:parallel}
A.~Rao.
\newblock Parallel repetition in projection games and a concentration bound.
\newblock In \emph{Proc. 40th ACM Symp. on Theory of Computing}, pages 1--10.
  2008.

\bibitem[Raz98]{Raz98}
R.~Raz.
\newblock A parallel repetition theorem.
\newblock \emph{SIAM J. Comput.}, 27(3):763--803, 1998.

\bibitem[Raz08]{Raz08}
R.~Raz.
\newblock A counterexample to strong parallel repetition.
\newblock In \emph{Proc. 49th IEEE Symp. on Foundations of Computer Science},
  pages 369--373. 2008.

\bibitem[Tre08]{Trev05}
L.~Trevisan.
\newblock Approximation algorithms for unique games.
\newblock \emph{Theory of Computing}, 4(1):111--128, 2008.

\bibitem[VB96]{vandenberghe&boyd:sdp}
L.~Vandenberghe and S.~Boyd.
\newblock Semidefinite programming.
\newblock \emph{SIAM Rev.}, 38(1):49--95, 1996.

\bibitem[WW01]{werner01:_bell}
R.~F. Werner and M.~M. Wolf.
\newblock Bell inequalities and entanglement.
\newblock \emph{Quantum Information and Computation}, 1(3):1--25, 2001.

\end{thebibliography}
\end{document}